\documentclass[a4paper]{IEEEtran}
\usepackage{amsmath}
\usepackage{amsfonts}
\usepackage{amssymb}
\usepackage{amsthm}
\usepackage{stmaryrd} % shortuparrow
\usepackage{enumitem}
\usepackage{tikz}
\usetikzlibrary{decorations.pathmorphing,shapes}
\usepackage[hyphens,spaces]{url}
\usepackage[hidelinks,bookmarks=false]{hyperref}
\hypersetup{
  pdfauthor={E.P.Csirmaz and L.Csirmaz},
  pdftitle={Synchronizing Many Filesystems in Near Linear Time, v2},
  bookmarksopen=false,
  pdfkeywords={file synchronization, algebraic model, 
  optimistic synchronization, linear complexity},
}
\newtheorem{theorem}{Theorem}

\newtheorem{claim}[theorem]{Claim}

\newtheorem{proposition}[theorem]{Proposition}

\theoremstyle{definition}
\newtheorem{defaux}[theorem]{Definition}
\newenvironment{definition}%
   {\pushQED{\hfill$\blacktriangleleft$}\defaux}%
   {\popQED\enddefaux}
\newtheorem{algaux}{Algorithm}
\newenvironment{algorithm}%
   {\pushQED{\hfill$\lozenge$}\algaux}%
   {\popQED\endalgaux}

\usepackage{newfloat}
\usepackage[noEnd=false,indLines=true,beginComment=~,commentColor=black!90,
   beginLComment=--~,endLComment=]{algpseudocodex}
\tikzset{algpxIndentLine/.style={draw,line width=0.5pt}}
\makeatletter
\algnewcommand\algorithmicforeach{\textbf{for each}}
\algdef{S}[FOR]{ForEach}[1]{%
        \algpx@startIndent\algpx@startCodeCommand\algorithmicforeach\ #1\ \algorithmicdo%
}
\pretocmd{\ForEach}{\algpx@endCodeCommand}{}{}
\algdef{SE}[LOOP]{Loop}{EndLoop}{%
       \algpx@startIndent\algpx@startCodeCommand\algorithmicloop%
}{%
        \algpx@endIndent\algpx@startEndBlockCommand\algorithmicend%
}
\algdef{SE}[IF]{If}{EndIf}[1]{%
        \algpx@startIndent\algpx@startCodeCommand\algorithmicif\ #1\ \algorithmicthen%
}{%
        \algpx@endIndent\algpx@startEndBlockCommand\algorithmicend%
}
\algdef{SE}[FOR]{For}{EndFor}[1]{%
        \algpx@startIndent\algpx@startCodeCommand\algorithmicfor\ #1\ \algorithmicdo%
}{%
        \algpx@endIndent\algpx@startEndBlockCommand\algorithmicend%
}
\makeatother
\algrenewcommand\alglinenumber[1]{\footnotesize #1~}% no colon after these numbers
\DeclareFloatingEnvironment[]{code}
\newenvironment{pseudocode}[1]
{\begin{code}[!thb]%
  \hrule\vskip -2pt \caption{\fontsize{9}{10}\selectfont #1}%
  \vskip 3pt
  \hrule\hbox{}\begin{algorithmic}[1]\ignorespaces}
{\end{algorithmic}
  \vskip 3pt \hrule
  \end{code}}

\usepackage{stmaryrd} % shortuparrow
\makeatletter
\def\correcthref{\hyper@anchor{\@currentHref}}
\makeatother
\newcommand\FS{\Phi} % filesystem
\newcommand\FSa{\Psi} % alternate filesystem
\newcommand\parent{{\shortuparrow}} % parent function
\newcommand\setm{\smallsetminus}
\newcommand\N{\mathbb N} % namespace
\newcommand\E{\mathbb O} % empty
\newcommand\F{\mathbb F} % file
\newcommand\D{\mathbb D} % directory
\newcommand\f{\mathsf f} % a file
\newcommand\type{\qopname\relax o{\mathsf{tp}}}
\newcommand\up{\qopname\relax o{\mathtt{up}}}
\newcommand\pt{\hbox{\Large.}}% huge dot in the pseudocode
\newcommand\orel{\ll} % order relation
\newcommand\subc{\Subset} % can be moved to the front
\renewcommand\bullet{\hbox{\textbullet}}
\let\preceq\preccurlyeq  % above or equal 
\DeclareFontFamily{U}{mathb}{\hyphenchar\font45}
\DeclareFontShape{U}{mathb}{m}{n}{
      <5> <6> <7> <8> <9> <10> gen * mathb
      <10.95> mathb10 <12> <14.4> <17.28> <20.74> <24.88> mathb12
      }{}
\DeclareSymbolFont{mathb}{U}{mathb}{m}{n}
\DeclareFontSubstitution{U}{mathb}{m}{n}
\DeclareMathSymbol\semle{3}{mathb}{"84}  % semantically inferior
\DeclareMathSymbol\semge{3}{mathb}{"85}
\renewcommand\setminus{\mathbin{\mskip-0.6\thinmuskip\smallsetminus\mskip-0.6\thinmuskip}}
\def\(#1,#2,#3){\langle #1,#2,#3\rangle} % internal command
\newcommand{\rot}[3]{#3#2#1} %
\newenvironment{itemz}[1][3pt]{%
\setitemize{topsep=#1,noitemsep,leftmargin=1.5\parindent,labelwidth=\parindent,labelsep=3pt,align=parleft}%
\begin{itemize}}{\end{itemize}}
\def\textsf#1{\hbox{$\xsf#1\endxsf$}}
\def\xsf#1#2\endxsf{%
  \ifx#1/\mathsf{\mkern-1mu/\mkern-1mu}\else\mathsf{#1}\fi
  \ifthenelse{\equal{#2}{}}{}{\xsf#2\endxsf}}

\newenvironment{Keywords}[1]{\xdef\IEEEkeywordsname{#1}\IEEEkeywords}{\endIEEEkeywords}
\title{\bfseries\fontsize{17.4}{15}\selectfont
Synchronizing Many Filesystems in Near Linear Time
}
\author{\fontsize{12.4}{12}\selectfont
Elod P. Csirmaz\IEEEauthorrefmark1
\thanks{\IEEEauthorrefmark1e-mail: \rot{\rot{maz.}{csir}{ep}com}{@}{elod}}
and
Laszlo Csirmaz\IEEEauthorrefmark2\thanks{\IEEEauthorrefmark2e-mail: csirmaz@renyi.hu
\space R\'enyi Institute, Budapest, and UTIA, Prague.}}

\begin{document}
\maketitle

\begin{abstract}
Finding a provably correct subquadratic synchronization algorithm for many
filesystem replicas is one of the main theoretical problems in Operational
Transformation (OT) and Conflict-free Replicated Data Types (CRDT)
frameworks.
Based on the Algebraic Theory of Filesystems, which incorporates
non-commutative filesystem commands natively, we developed and built a
proof-of-concept implementation of an algorithm suite which synchronizes an
arbitrary number of replicas. The result is provably correct, and the
synchronized system is created in linear space and time after an initial
sorting phase. It works by identifying conflicting command pairs and
requesting one of the commands to be removed. The method can be guided to
reach any of the theoretically possible synchronized states.
The algorithm also allows asynchronous usage. After the client sends a
synchronization request, the local replica remains available for further
modifications. When the synchronization instructions arrive, they can be
merged with the changes made since the synchronization request. The suite
also works on filesystems with directed acyclic graph-based path structure
in place of the traditional tree-like arrangement. Consequently, our
algorithms apply to filesystems with hard or soft links as long as the links
create no loops.

\begin{Keywords}{Keywords}
file synchronization;
algebraic model; optimistic synchronization; linear complexity.
\end{Keywords}

\begin{Keywords}{MSC classes} %
08A02,
08A70, %
68M07, %
%68M14, %
68P05. %
%68P20 %
\end{Keywords}

\begin{Keywords}{ACM classes} %
D.4.3, %
E.5, %
F.2.2, %
G.2. %
\end{Keywords}
\end{abstract}

% -----------------------------------------------------------------------

\section{Introduction and related works}

Synchronizing diverged copies of some data stored on a variety of devices
and/or at different locations is an ubiquitous task. The last two decades
saw a proliferation of practical and theoretical works addressing this
problem.
According to \cite{techradar23}, ``file synchronization is a feature usually
included with backup software in order to make is easier to manage and
recover data as and when required.'' File synchronization usually delivered
through cloud services. Dedicated file synchronizing
solutions frequently come with additional tools not just for managing the saved
data, but also to allow for file sharing and collaboration with stored files
and documents.

These cloud storage services are easily accessible for the end-user because
the service front-ends are very well integrated into web clients as well as
desktop and mobile environments. Simple user interfaces hide the complex and
sophisticated service back-ends \cite{MOS2018}. Collaboration services are
frequently integrated into the ``cloud storage'' environment. For example,
Google Docs is an application layer integrated into Google Drive storage,
Office 365 is integrated with One Drive storage and Dropbox Paper service is
an extension of Dropbox storage.

To address the emerging challenges in the more specific fields of
distributed data storage and collaborative editors two competing theoretical
frameworks have emerged: Operational Transformation (OT) and Conflict-free
(or Commutative) Replicated Data Types (CRDT). OT appeared in the seminal
work of \cite{SJZ98}, and was refined later, among others, in
\cite{SE98,SLL11,NS16}. The main applications are collaborative editors, the
most notable example being Google Docs \cite{DR10}. CRDT emerged as an
alternative with a stronger theoretical background, see
\cite{CRDT-orig,CRDT-overview}. Both OT and CRDT have been applied
successfully in a variety of synchronization tasks, including file
synchronizers \cite{NSC16,TSR15,LI21}. Finding a provably correct,
subquadratic synchronization algorithm, however, has remained one of the
main open problems both in OT and CRDT \cite{syncpal-thesis}.

The Algebraic Theory of Filesystems \cite{Csi16,CC22a} introduced algebraic
manipulations of filesystem commands, and provided the foundations for
automated checking of certain filesystem properties. It is reminiscent of
both OT \cite{NS16} and CRDT inasmuch as instead of pure traditional
filesystem commands it uses operations enriched with contextual information.
While these operations are not fully commutative -- as would be requested by
CRDT \cite{SPMB11} --, the non-commutative parts can be isolated
systematically and handled separately. In \cite{CC22a} this framework has
been used successfully to create a provably correct theoretical filesystem
synchronizer for two replicas together with a complete analysis of all
possible synchronized states. The present work extends \cite{CC22a}
significantly by
\begin{itemz}
\item[--] providing the theoretical foundation for synchronizing an
\emph{arbitrary number} of replicas;

\item[--] developing, for the first time, a provably correct synchronization
algorithm which works in \emph{linear time} after an initial sort
thus in subquadratic total running time;

\item[--] allowing \emph{asynchronous usage}, namely, after requesting
synchronization the local replicas need not be locked;

\item[--] allowing for \emph{late comers}, when a replica can be upgraded to the
synchronized state
without providing the local changes;

\item[--] generalizing the traditional tree-like filesystem skeleton to
arbitrary \emph{acyclic graphs}, thus extending the applicability of
the synchronization algorithm.
\end{itemz}
In this paper we will use the term \emph{near linear} to mean ``linear up to
a logarithmic factor.'' Thus sorting requires near linear time
\cite{knuth97}, and the total time required by the above synchronization
algorithm is also near linear.

\begin{figure}[htb]\centering
\begin{tikzpicture}[decoration={zigzag,amplitude=1.5pt,segment length=1.2mm,pre=lineto,pre length=3pt}]
\clip (-1.5,-1.24) rectangle (5.7,2.54);
\fill[fill=blue!8] (-0.55,-0.3) rectangle (-0.05,1.9);
\fill[fill=blue!8] (2.95,-0.3) rectangle (3.45,1.9);
\fill[fill=red!8] (0.9,-0.3) rectangle (1.4,1.9);
\fill[fill=red!8] (4.38,-0.3) rectangle (4.90,1.9);
\foreach\y/\z in {0/{\FS_1},0.8/{\FS_2},1.6/{\FS_3}}{
\draw[<-,decorate] (0.92,\y) -- (-0.125,\y);
\draw (-0.3,\y) node{$\FS$}; \draw (1.166,\y-0.02) node {$\z$};
\draw[dotted] (1.45,\y) -- (2.7,\y); 
\draw (2.85,\y-0.03) node {$\rightsquigarrow$};
\draw (3.22,\y) node{$\FSa$};
{\let\FS\FSa
\draw[<-,decorate] (4.42,\y) -- (3.38,\y);\draw (4.65,\y-0.02) node {$\z$};
}
\draw[dotted] (4.9,\y) -- (5.7,\y);
}

\draw (-1.1,0) node {\footnotesize copy1};
\draw (-1.1,0.8) node {\footnotesize copy2};
\draw (-1.1,1.6) node {\footnotesize copy3};

\draw (-0.5,2.2) node{\footnotesize synchronized};
\draw[blue,->] (-0.8,2.05) to[out=270,in=160] (-0.3,1.3);
\draw (4.0,2.3) node{\footnotesize synchronized};
\draw[blue,->] (3.7,2.15) to[out=270,in=20] (3.3,1.3);
\draw (1.2,2.4) node{\footnotesize local\,changes};
\draw[->] (0.68,2.25) to[out=250,in=90] (0.5,1.7);
\draw[->] (0.85,2.25) to[out=245,in=90] (0.5,0.9);
\draw[->] (1.05,2.25) to[out=240,in=90] (0.6,0.1);
\draw (2.5,2.1) node{\footnotesize diverged};
\draw[red!50!black,->] (2.4,1.95) to[out=230,in=15] (1.25,1.3);

\foreach\y in {-0.2}{
\draw (1.25,\y-0.63) rectangle (2.85,\y-1.03);
\draw (2.05,\y-0.85) node {\footnotesize algorithm};
\draw[red!50!black,->] (1.4,1.6) to[out=-45,in=90] (1.9,\y-0.6);
\draw[red!50!black,->] (1.4,0.8) to[out=-40,in=90] (1.82,\y-0.6);
\draw[red!50!black,->] (1.4,0.0) to[out=-30,in=90] (1.75,\y-0.6);

\draw[blue,<-] (2.8,1.6-0.05) to[out=225,in=90] (2.22,\y-0.6);
\draw[blue,<-] (2.8,0.8-0.05) to[out=220,in=90] (2.3,\y-0.6);
\draw[blue,<-] (2.8,0.0-0.05) to[out=210,in=90] (2.38,\y-0.6);

\draw (1.1,\y-0.48) node {\footnotesize requests};
\draw (3.1,\y-0.44) node {\footnotesize commands};
%%%%%%%%%%%%%%%
\draw[red!50!black,->] (4.9,1.6) to[out=-45,in=90] (3.5+1.9,\y-0.6);
\draw[red!50!black,->] (4.9,0.8) to[out=-40,in=90] (3.5+1.82,\y-0.6);
\draw[red!50!black,->] (4.9,0.0) to[out=-30,in=90] (3.5+1.75,\y-0.6);
\draw (3.5+1.25,\y-0.63) rectangle (3.5+2.85,\y-1.03);
\draw (3.5+2.05,\y-0.85) node {\footnotesize algorithm};
}

\end{tikzpicture}

\caption{\fontsize{9}{10}\selectfont
The synchronization cycle. Identical copies of the same filesystem
are edited independently. Each replica sends the locally created update
information to the synchronizer, which returns the commands to be executed
on the local copy to update it to a common synchronized state.
}\label{fig:sync1}

\end{figure}
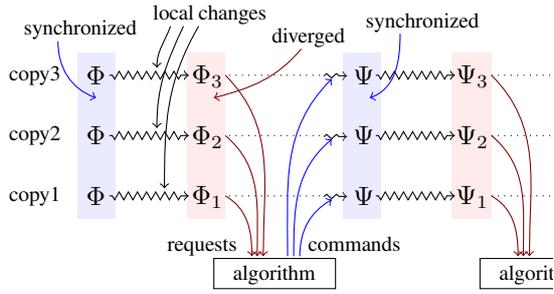

This paper follows the traditional paradigm of filesystem synchronization
described in e.g. \cite{BP98} and illustrated in
Figure \ref{fig:sync1}. The starting point is a set of identical copies of the same filesystem,
possibly stored
at different locations,
on different hardware and architectures (cloud servers, mobile devices, laptop and desktop
computers with various operating systems)
or using different software implementations (e.g. ext4, btrfs, ZFS, NTFS, APFS,
or database file systems).
Each of these replicas is edited (modified) locally. At a
certain time the diverged copies call for \emph{synchronization} by sending
a description of the diverged state to a central server. After
receiving the requests, the server computes filesystem commands
which transform each replica into a common synchronized state, and
sends them back to the replicas. The replicas execute the received
commands on their local copy transforming all diverged copies into a
new identical synchronized state. At that point the synchronization cycle
can start again.

Synchronizers typically require locking the replicas during the whole
synchronization process, meaning that no modifications are allowed after the
synchronization request is sent. (The locked time period is
indicated by the dotted lines on Figure \ref{fig:sync1}.)
\emph{Asynchronous}, or \emph{optimistic} synchronization allows additional
local modifications after the synchronization request is sent as
depicted on the top of Figure \ref{fig:sync2}. When the synchronization
commands arrive from the server, those commands are modified to reflect the
additional changes, and then applied to the replica. The result should
be the same as when performing synchronization without the additional changes, and
then applying them to the synchronized filesystem afterwards---as indicated at the bottom
of the figure.

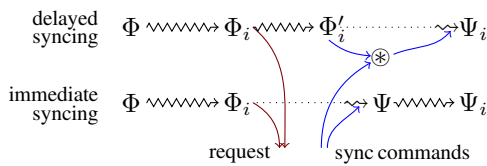
\begin{figure}[htb]\centering
\begin{tikzpicture}[decoration={zigzag,amplitude=1.5pt,segment length=1.2mm,pre=lineto,pre length=3pt}]
\foreach\y in {0,1.0}{
\draw (-0.2,\y) node{$\FS$};
\draw (1.2,\y-0.02) node {$\FS_i$};
\draw (4.3,\y-0.02) node {$\FSa_i$};
\draw[<-,decorate] (0.95,\y) -- (0.0,\y);
}

\draw (-0.5,1.1) node[left] {\footnotesize delayed}
      (-0.5,0.85) node[left] {\footnotesize syncing}
      (-0.5,0.12) node[left] {\footnotesize immediate}
      (-0.5,-0.16) node[left] {\footnotesize syncing};

\draw[red!50!black,->] (1.4,1.0) to[out=-40,in=90] (1.82,-0.6);
\draw[red!50!black,->] (1.4,0.0) to[out=-30,in=90] (1.75,-0.6);

\draw[<-,decorate] (2.2,1.0) -- (1.4,1.0);
\draw (2.45,1.0-0.02) node {$\FS'_i$};

\draw[dotted] (1.45,0)--(2.65,0);
\draw (2.76,-0.03) node {$\rightsquigarrow$};
\draw[dotted] (2.55,1.0)--(3.8,1.0);
\draw (3.950,1.0-0.03) node {$\rightsquigarrow$};
\draw (3.1,0) node {$\FSa$};
\draw[<-,decorate] (4.05,0) -- (3.25,0);

\draw (3.08,0.6) node {$\oast$};
\draw[blue,<-] (2.95,0.53) to[out=210,in=90] (2.3,-0.6);
\draw[blue,<-] (4.0,1.0-0.05) to[out=230,in=23] (3.25,0.64);
\draw[blue,<-] (2.95,0.62) to[out=175,in=-40] (2.4,0.83);
\draw[blue,<-] (2.8,0-0.05)   to[out=225,in=90] (2.38,-0.6);

\foreach\y in{-0.2}{
\draw (1.23,\y-0.48) node {\footnotesize request};
\draw (3.37,\y-0.48) node {\footnotesize sync\,commands};
}
\end{tikzpicture}

\caption{\fontsize{9}{10}\selectfont
Asynchronous synchronization. Top line: after the synchronization
request has been sent, additional local modifications are made to the
filesystem. When receiving the synchronization commands, they are
modified using the current state of the local filesystem. 
Bottom line: the end result should be the
same as applying the synchronization commands immediately and then making
the local modifications afterwards.
}\label{fig:sync2}

\end{figure}

The main focus of this paper is filesystem synchronization, or the
synchronization of data stored in the nodes of a tree or directed acyclic
graph. The stored data is considered to be an indivisible unit, and the task
of consolidating different versions of the same data is not considered. It
should be solved by other methods specifically tailored to this task.

Some practical aspects of data synchronization are not touched and are out
of the scope of this paper. Managing user access and permissions, when and
how to allow file sharing and collaboration are especially important due to
security considerations \cite{techradar23}. Additionally, file
synchronization should use very strict security protocols to ensure that
data is safely protected and secured at all times and to make data leaks and
malicious access less likely.

Synchronizers should also minimize network traffic. Unlike the popular Rsync
utility \cite{rsync} available for comparing and synchronizing files, there
is currently no such a ``middleware'' utility for general data sets
\cite{BOS2023}. Hopefully our work is a small step in that direction.

The rest of this paper is organized as follows. Section \ref{sec:defs}
recalls the building blocks of the Algebraic Theory of Filesystems including
the filesystem model, the augmented filesystem commands and their basic
properties. This section does not contain new results and its purpose is to
give the reader a comprehensive summary of the topic on which the rest of
the paper relies. For more intuition, explanation and examples please
consult \cite{CC22a}. Section \ref{sec:overview} is a high-level overview of
how filesystem synchronization can be handled in the algebraic framework.
This section \emph{defines} what constitutes a synchronized state of several
diverged replicas, rather than providing a method of creating it. The
definition automatically guarantees many desired and required properties of
the merged filesystem, an indication of the strength and adequacy of the
algebraic framework. It is discussed how all synchronized states can be
achieved by conflict resolution, paving the way towards the near linear
synchronization algorithm discussed in Sections \ref{sec:algorithms} and
\ref{sec:all-mergers}. Asynchronous (optimistic) synchronization is
discussed at the end of Section \ref{sec:overview}.

The base algorithms of the synchronization suite are discussed in Section
\ref{sec:algorithms} including the one which generates the command sequence
(called \emph{merger}) which produces the merged state in subquadratic time.
Supporting theoretical results are collected and proved in Section
\ref{sec:theory}. Section \ref{sec:all-mergers} discusses how all possible
synchronized states can be generated by a nondeterministic algorithm running
in linear time. Section \ref{sec:results} presents some empirical results
justifying the claims about the running time of the algorithms. Finally,
Section \ref{sec:conclusion} concludes the paper with some extensions and
open problems. Most notably, our algorithms, with some modification, work
not only on the tree-like filesystem skeletons as stipulated by the
algebraic theory, but also on filesystems based on directed acyclic graphs.

This work focuses mainly on algorithmic aspects so many theoretical
justifications are deliberately phrased in general terms. Rigorous proofs
would require substantially more space and, in our opinion, would not
provide additional insight.
A proof-of-concept implementation of the algorithms presented in this paper
in Python can be found at
\url{https://github.com/csirmaz/algebraic-reconciler}.

% -----------------------------------------------------------------------
\section{Definitions}\label{sec:defs}

This section recalls the notions and basic results of the Algebraic Theory
of Filesystems \cite{Csi16} with some illustration of the concepts. The main
ingredient is a highly symmetric set of filesystem commands which are
enriched with contextual information. Devising such a command set which is
amenable to algebraic manipulation was one of the main contributions of
\cite{Csi16}. For more intuition and explanation on this see \cite{CC22a}.

% -------------------
\subsection{Filesystems}\label{subsec:filesystem}

The filesystem model reflects the most important high-level
aspects of real-word filesystems. It is a mixture of
identity- and path-based models \cite{syncpal-thesis,TSR15}. The contents of
the filesystem are stored at \emph{nodes} which are identified, or labeled,
by a set of fixed and predetermined paths. The collection of all (virtually)
available nodes is fixed in advance
(while in a real-world filesystem only a
restricted subset of those paths is present).
No path operations are
considered; in particular, our model does not support the creation or
deletion of
\emph{links}. In the basic case no links are allowed at all, thus the
namespace---the set of available nodes or paths---forms a collection of
rooted trees. An actual filesystem populates this fixed namespace with
values. If $\FS$ is a filesystem, the value stored at node $n$ is
denoted by $\FS(n)$. Valid filesystems are required to have the
\emph{tree-property} all the time, meaning that along any branch starting
from a root node there must be zero or more \emph{directories}, zero or one
\emph{file}, followed by \emph{empty} nodes only. If $\FS$ does not have the
tree-property, then we say that $\FS$ is \emph{broken}. 

Formally, the paths form a forest-like \emph{namespace:} a set $\N$ endowed
with the partial function $\parent:\N\to\N$ returning the parent of every
non-root node, while this function is not defined on roots. If $n=\parent m$
then $n$ is the parent of $m$, and $m$ is a child of $n$. For two nodes
$n,m\in\N$ we say that $n$ is above $m$, or $n$ is an ancestor of $m$, and
write $n\prec m$, if $n=\parent^i m$ for some $i\ge 1$. As usual, $n\preceq
m$ denotes $n\prec m$ or $n=m$. As the parent function $\parent$ induces a
tree-like structure on $\N$, the relation $\preceq$ is a partial order. Two
nodes $n,m\in \N$ are \emph{comparable} if either $n\preceq m$ or $m\preceq
n$, and they are \emph{uncomparable} or \emph{independent} otherwise.

In practice, nodes of the filesystem are labeled by \emph{complete paths},
where directory names are separated by the slash character. Thus the root
has the label \textsf{/}, nodes \textsf{/a}, \textsf{/xxx} are on the first
level just under the root, and \textsf{/a/b/cc/d} is (the label of) a node
on level four whose ancestors are \textsf{/a/b/cc}, \textsf{/a/b},
\textsf{/a}, and \textsf{/}. The implementation of the algorithms also
follows this convention.

As indicated above, the value stored by a filesystem at a node can be a
\emph{directory}, can be \emph{empty} or can be a \emph{file}. The type of
this value is denoted by $\D$, $\E$ and $\F$, respectively, corresponding to
these possibilities. With an abuse of notation we also use $\D$ for the
directory value, and $\E$ for the empty value (where $\E$ means ``no
value'', not to be confused with a file which has no content). When the node
value is a file, then the node stores the complete file content (including
the possibility that this content is empty). While this filesystem model
allows only one directory value and only one empty value (but see the
discussion in Section \ref{sec:conclusion} on relaxing this limitation),
there are many different possible file values of type $\F$ representing
different file contents. The value types are \emph{ordered} with $\E$ being
the lowest, and $\D$ being the highest, written as $\E < \F < \D$. The type
of the filesystem value $x$ is denoted by $\type(x)$.

% -------------------
\subsection{Filesystem commands}

Real-life filesystems are usually manipulated by commands like creating or
deleting files and directories, modifying (editing, appending to) a file, or
moving an existing file or directory to another location. Our model contains
similar commands, but with some modifications; the first of which is that we
only consider commands which affect the filesystem at a single node. Thus a
\emph{move} command should be represented as a sequence of a \emph{delete}
and a \emph{create}.

Second, commands in our model include the complete new value to be stored in
the filesystem. Even if file contents are just partially modified or are
appended to, the full new value must be supplied. This allows our model to
use an unified representation of all single-node commands, as they can be
fully specified by the node (path) at which the command acts and the new
value (including a directory and empty value) to be stored there.

Third, as was observed in \cite{Csi16}, enriching filesystem commands with
additional contextual information--- in this case, the previous content at
the affected node---makes them amenable to algebraic manipulation.

\begin{definition}[Filesystem commands]\label{def:commands}
A filesystem command is a triplet $\sigma=\(n,x,y)$, where $n\in \N$ is the
node on which $\sigma$ acts, $x$ is the content at node $n$ before
$\sigma$ is executed (the contextual information, precondition), and $y$ is
the new content.
\end{definition}

It is clear that every real-life filesystem command acting on a single node
can be easily (and automatically) transformed into this internal
representation. For example, \emph{rmdir}$(n)$ corresponds to $\(n,\D,\E)$,
which replaces the directory value at $n$ by the empty value. The command
$\(n,\E,\D)$ creates a directory at $n$, but only if the node $n$ has no
content, that is, there is no directory or file at $n$ (a usual requirement
when creating a directory). This command reflects the usual behavior or
\emph{mkdir}. For files $\f_1$ and $\f_2\in\F$ the command $\(n,\f_1,\f_2)$
replaces $\f_1$ stored at $n$ by the new content $\f_2$. This latter command
can be considered to be an equivalent of \emph{edit}$(n,\f_2)$.

As an example, creating a copy of the file \textsf{/home/user/text} in the
same directory under the name ``\textsf{copy}'' and then deleting the
original file is represented by the sequence of commands
$$
\(\textsf{/home/user/copy},\E,\f_o)  \mbox{~~and~~}
\(\textsf{/home/user/text},\f_o,\E),
$$
where $\f_o$ is the file content at the ``\textsf{text}'' node.

Applying the command $\sigma$ to a filesystem $\FS$ is written as the left
action $\sigma\FS$. The command $\sigma=\(n,x,y)$ is \emph{applicable} to
$\FS$ if $\FS$ contains $x$ at node $n$, that is, $\FS(n)=x$ (the
precondition holds), and after changing the content at $n$ to $y$ the
filesystem still has the tree property. If $\sigma$ is not applicable to
$\FS$, then we say that $\sigma$ \emph{breaks the filesystem}. If $\sigma$
does not break $\FS$, then $\sigma$ is applicable to $\FS$, and its
execution changes $\FS$ at node $n$ only.

Command sequences are applied from left to right, thus $(\sigma\alpha)\FS =
\alpha(\sigma\FS)$, where $\alpha$ is a command sequence. Composition of
sequences is written as the concatenation $\alpha\beta$, but occasionally we
write $\alpha\circ\beta$ to emphasize that $\beta$ is to be executed after
$\alpha$. A sequence breaks a filesystem if one of its commands breaks the
filesystem when it was to be applied. The sequence $\alpha$ is
\emph{non-breaking} if there is at least one filesystem $\alpha$ does not
break; otherwise it is \emph{breaking}.

Sequences $\alpha$ and $\beta$ are \emph{semantically equivalent}, written
as $\alpha\equiv\beta$, if they have the same effect on all filesystems,
that is, $\alpha\FS=\beta\FS$ for all $\FS$. We write $\alpha\semle\beta$ to
denote that $\beta$ semantically extends $\alpha$, that is,
$\alpha\FS=\beta\FS$ for all filesystems that $\alpha$ does not break.

For example, the sequence which creates a file at some node $n$ and then
changes this file to a directory is equivalent to the single command which
creates the diretory directly:
$$
  \(n,\E,\f) \circ \(n,\f,\D) \equiv \(n,\E,\D),
$$
while creating the file and deleting it immediately is  semantically
strictly weaker than applying the ``null'' command $\(n,\E,\E)$, thus we
only have
$$
   \(n,\E,\f) \circ \(n,\f,\E) \semle \(n,\E,\E).
$$
This is because the right hand side is applicable when the parent of $n$
contains a file, while the left hand side would break such a filesystem.

The \emph{inverse} of $\sigma=\(n,x,y)$ is $\sigma^{-1}=\(n,y,x)$. For a
sequence $\alpha$ its inverse $\alpha^{-1}$ consists of the inverses of the
commands in $\alpha$ in reverse order. The inverse has the expected
property: if $\alpha$ does not break $\FS$, then $(\alpha^{-1}\alpha)\FS =
\FS$, that is, $\alpha^{-1}$ rolls back the effects of $\alpha$. Observe
that $\alpha$ is non-breaking if and only if so is $\alpha^{-1}$.

% -------------------
\subsection{Command types, execution order}

The input and output \emph{values} of $\sigma=\(n,x,y)$ are $x$ and $y$,
respectively, while the input and output \emph{types} are $\type(x)$ and
$\type(y)$. Commands are classified by their input and output types using
\emph{patterns}. The command $\(n,x,y)$ matches the pattern $\(n,\mathcal
P_x,\mathcal P_y)$ if $\type(x)$ is listed in $\mathcal P_x$, and $\type(y)$
is listed in $\mathcal P_y$. In a pattern the symbol $\bullet$ matches any
value. As an example, every command matches $\(\bullet,\E\F\D,\D\F\E)$.

Commands with identical input and output values are \emph{null commands}.
Null commands do not change the filesystem (but can break it if the
precondition does not hold). \emph{Structural commands} change the type of
the stored data. Structural commands are further split into
\emph{constructors} and \emph{destructors}. A constructor increases the type
of the stored value, while a destructor decreases it. Thus a constructor
matches either $\(\bullet,\E,\F\D)$ or $\(\bullet,\F,\D)$, and a destructor
matches either $\(\bullet,\D\F,\E)$ or $\(\bullet,\D,\F)$. Observe that
$\sigma$ is a constructor if and only if $\sigma^{-1}$ is a destructor.
Finally, non-null commands matching $\(\bullet,\F,\F)$ are \emph{edits}.

The binary relation $\sigma\orel\tau$ between commands on parent--child
nodes captures the notion that $\sigma$ must precede $\tau$ in the execution
order.

\begin{definition}[$\orel$ relation, $\orel$-chain]\label{def:exec-order}
The relation $\sigma\orel\tau$ holds if the pair matches either $\(n, \D\F,
\E)\orel(\parent n,\D,\F\E)$, or matches $\(\parent n,\E\F, \D) \orel \(n,
\E,\F\D)$. An $\orel$-chain is a sequence of $\orel$-related commands
connecting its first and last element.
\end{definition}

The first case in Definition \ref{def:exec-order} corresponds to the
requirement that before deleting a directory its descendants should be
deleted. The second case says that a file or directory can only be created
under an existing directory. Observe that $\sigma\orel\tau$ if and only if
$\tau^{-1}\orel\sigma^{-1}$, and in this case either both $\sigma$ and
$\tau$ are constructors, or both are destructors.

% -------------------
\subsection{Canonical sets and sequences}

Commutativity is a core concept in command-based synchronization
\cite{BP98}, where, in fact, the task is to determine in what order (and
which) modifications made to other replicas can be applied to a particular
replica. If two modifications or commands commute, that is, their result
does not depend on the order in which they are applied, then they do not
represent conflicting updates to the data, as they can be seen as
independent. Unsurprisingly, then, commutativity plays a central role in
CRDT (see \cite{CRDT-orig,CRDT-overview}), where basic data types with
special operators are devised so that executing the operators in different
orders yields the same results.

While not all filesystem commands commute, non-commutative pairs can be
isolated systematically. If $\sigma$ and $\tau$ are on different nodes which
are not in parent--child relation, then they commute ($\sigma\tau$ and
$\tau\sigma$ are semantically equivalent). An example of this is creating a
file in some directory and editing another existing file. The affected nodes
where the changes are made are independent. If $\sigma$, $\tau$ are non-null
commands on parent--child nodes, then either $\sigma\tau$ breaks every
filesystem, or, necessarily, $\sigma\orel\tau$. If a command $\sigma$ on
node \textsf{/a/b/c} is followed immediately by another command $\tau$ on
the node \textsf{/a/b/c/d} successfully, then $\sigma$ must create a
directory at \textsf{/a/b/c} (which location previously was either empty, or
contained a file), thus \textsf{/a/b/c/d} was empty, so $\tau$ must create
either a file or a directory there. Consequently we have $\sigma\orel\tau$
by Definition \ref{def:exec-order}. Consecutive commands on the same node
either break every filesystem (if the second command requires a different
value than the output of the first command), or can be replaced by a single
command (while extending the semantics). These easy facts imply some strong
and intricate structural properties of non-breaking command sequences.
Exploring and using these properties made possible the complete and thorough
investigation of the synchronization process of two diverged replicas in
\cite{CC22a}, as well as devising the first provably correct subquadratic
synchronization algorithm in this paper.

Intuitively, a canonical sequence is just the ``clean'' version of a
non-breaking command sequence. An important property of canonical sequences
is that their semantics is determined uniquely by the set of commands they
contain, see Theorem \ref{thm:basic}. Command sets that can be arranged into
canonical sequences are also called canonical. For the formal definitions we
need two more notions. A command sequence $\alpha$ \emph{honors $\orel$}, if
for any two commands $\sigma,\tau\in\alpha$, $\sigma$ precedes $\tau$ in the
sequence whenever $\sigma\orel\tau$. The command set $A$ is
\emph{$\orel$-connected} if for any two commands $\sigma,\tau\in A$, if
$\sigma$ and $\tau$ are on different comparable nodes, then they are
connected by an $\orel$-chain (Definition \ref{def:exec-order}) consisting
of commands in $A$. In particular, if $A$ is $\orel$-connected and $\sigma,
\tau\in A$ are on the comparable nodes $n$ and $m$, then $A$ has commands on
each node between $n$ and $m$.

\begin{definition}[Canonical sets and sequences]\label{def:canonical}
The command set $A$ is canonical, if the following three conditions hold:
\begin{itemz}[1pt]
\item $A$ does not contain null-commands;
\item $A$ contains at most one command on each node; and
\item $A$ is $\orel$-connected.
\end{itemz}
The command sequence $\alpha$ is a canonical if the commands in $\alpha$
form a canonical set and, additionally, $\alpha$ honors $\orel$.
\end{definition}

For example, the set consisting of the three commands
$\(\textsf{/a/b},\D,\f_5)$, $\(\textsf{/a/c},\f_o,\E)$, and
$\(\textsf{/a/c},\f_5,\f_5)$ is not canonical for two reasons: the third
command is a null command, and the second and the third commands are on the
same node. Neither is canonical the command set
$$
\{\,
\(\textsf{/a/b/c/d},\f_s,\E),
\(\textsf{/a},\D,\f_s),
\(\textsf{/a/b},\D,\E) \,\}
$$
as its first and last elements are not $\orel$-connected. (The second and
third elements are $\orel$-connected.) Adding the command
$\(\textsf{/a/b/c},\D,\E)$ to this set makes it canonical.

% use \correcthref for hyperref broken by {itemz}
\begin{theorem}[E.P.~Csirmaz, \cite{Csi16}]\correcthref\label{thm:basic}

\begin{itemz}[0pt]
\item[]
\item[a)] If two canonical sequences share the same command set, then they
are semantically equivalent. Actually, they can be transformed into each
other using commutativity rules.

\item[b)] Canonical sequences are non-breaking.

\item[c)] Every non-breaking sequence $\alpha$ can be transformed into a
canonical sequence $\alpha^*\semge\alpha$ (that is, $\alpha$ and $\alpha^*$
have the same effect on filesystems that $\alpha$ does not break, but
$\alpha^*$ might work on more filesystems).

\item[d)] Canonical sets can be ordered to honor $\orel$, that is, to 
become canonical sequences.
\qed
\end{itemz}
\end{theorem}

As an example, the only order honoring the $\orel$ relation of the canonical
set
\begin{equation}\label{eq:example}
E = \big\{\,
\(\textsf{/a/b/c/d},\f_s,\E),
\(\textsf{/a},\D,\f_s),
\(\textsf{/a/b},\D,\E),
\(\textsf{/a/b/c},\D,\E) \,\big\}
\end{equation}
is the one which executes them ``bottom up'' starting with the command on
\textsf{/a/b/c/d} and ending with the command on node \textsf{/a}.

Algorithm \ref{alg:canonical} in Section \ref{sec:algorithms} checks, in
near linear time, whether a command set $A$ is canonical. Algorithm
\ref{alg:ordering} arranges a canonical set into a canonical sequence.

By virtue of Theorem \ref{thm:basic}\,\emph{a)} and \emph{d)}, the semantics
of canonical sequences is determined by the unordered set of their commands,
even if their order cannot be recovered uniquely. (This happens when the set
contains commands on uncomparable nodes.) Thus when only the semantics is
concerned, a canonical sequence can, and will, be replaced by the set of its
commands. For example, when we write $A\FS$ where $A$ is a canonical command
set, we mean that commands in the set $A$ should be applied in some (or any)
$\orel$-honoring order to the filesystem $\FS$. Similarly, $A\circ B$ means
that first the commands in $A$ are applied in some $\orel$-honoring order,
followed by the commands in $B$, again in some $\orel$-honoring order.

The property that commands in a canonical set can be executed in different
orders while preserving its semantics is a variant of the
\emph{commutativity principle} of CRDT \cite{PMS09}. Definition
\ref{def:iseg} below discusses a special case of the reordering when a
subset of the commands is to be moved to the beginning of the execution
line.

\begin{definition}[Initial segment]\label{def:iseg}

For a canonical set $A$ we write $B\subc A$, and say that $B$ is an
\emph{initial segment of $A$}, to indicate that $B$ is not only a subset of
$A$, but can also be moved to the beginning of an ordering of $A$ while
keeping the semantics. In other words, $A\equiv B\circ (A\setm B)$.
\end{definition}

We remark that if $B$ is an initial segment of $A$, then both $B$ and
$A\setm B$ are canonical, see \cite[Proposition 5]{CC22a}. For example, the
canonical set $E$ in (\ref{eq:example}) has three proper initial segments:
\begin{itemz}
\item[] $\{\(\textsf{/a/b/c/d},\f_s,\E),
\(\textsf{/a/b},\D,\E),
\(\textsf{/a/b/c},\D,\E) \big\}$,
\item[] $\big\{\(\textsf{/a/b/c/d},\f_s,\E),
\(\textsf{/a/b/c},\D,\E) \big\}$ , and
\item[] $\big\{\(\textsf{/a/b/c/d},\f_s,\E) \big\}$
\end{itemz}
and, of course, all of them are canonical.

% -------------------
\subsection{Refluent sets}\label{subsec:refluent}

Following the terminology of \cite{CC22a}, canonical sets $A$ and $B$ are
called \emph{refluent} if there is at least one filesystem on which both of
them work (neither of them breaks). In general, the canonical sets $A_1,
\dots, A_k$ are \emph{jointly refluent} if there is a single filesystem on
which all of them work. It is clear that if $k$ canonical sets are jointly
refluent, then they are pairwise refluent as well. The converse statement,
that if the canonical sets $A_i$ are pairwise refluent, then they are also
jointly refluent, is stated and proved as Proposition \ref{prop:refluent} in
Section \ref{sec:theory}.
The concept of refluence arises naturally in file synchronization. Command
sequences representing changes to the replicas are refluent as they have
been applied to identical copies of the same filesystem.

% -----------------------------------------------------------------------
\section{Filesystem synchronization}\label{sec:overview}

The synchronization paradigm used in this paper follows the traditional one
described in, e.g., \cite{BP98} and depicted in Figure \ref{fig:sync1}. At
the beginning of the synchronization cycle each replica stores an identical
copy of the same filesystem $\FS$. Due to local modifications the replicas
diverge, and after some time the filesystem in the $i$-th replica changes to
$\FS_i$. At a certain moment the replicas call for synchronization by
sending information extracted by the \emph{update detector} (run locally) to
a central server which hosts the reconciliation algorithm. After the server
has received all update information, it determines what the common
synchronized filesystem $\FSa$ will be. Then it sends instructions to the
replicas separately telling them how to transform their local filesystem to
the synchronized one. Finally, each replica executes the received
instructions which transforms their local copy to the synchronized state
$\FSa$, optionally informing the user about some (or all) of the conflicts
and how they have been resolved.

% --------------------------------
\subsection{Update detector}\label{subsec:update-detector}

Depending on the data communicated by the replicas, synchronizers are
categorized as either \emph{state-based} or \emph{operation-based}
\cite{CRDT-overview,syncpal-thesis}. In state-based synchronization replicas
send the current state of their filesystems, or merely the differences
between their current state and the last known synchronized state
\cite{AK08}. Frequently the local copy does not have access to the original
synchronized state because of its limited resources, and transmitting the
whole current state is prohibitively expensive. It is an active research
area to devise efficient transmission algorithms which transmit the
differences only \cite{rsync,LWJ13,PCS18,FQL12}. Operation-based
synchronizers transmit the complete log (or trace) of all operations
performed by the user \cite{K10}. It has been observed that in practice
these logs are poorly maintained and are not always reliable
\cite{Qia04,ZDA13}.

Theorem \ref{thm:basic} suggests that the update information the replicas
send to the central server (or, rather, the data on which the synchronizer
algorithm works) can be a canonical command set which transforms the
original synchronized filesystem to the current replica. By Theorem
\ref{thm:state-based} below this set is not only a succinct representation
of the differences, but can also be generated in time proportional to the
size of the filesystems (by traversing the original $\FS$ and the modified
filesystem $\FS_i$ simultaneously).

\begin{theorem}[{\cite[Theorem 19]{CC22a}}]\label{thm:state-based}
Let $\FS$ and $\FSa$ be two filesystems. The command set
$$
    A_{\FS{\shortrightarrow}\FSa}=\{ \(n,{\FS(n)},{\FSa(n)}) : n\in \N 
   \mbox{ and } \FS(n)\neq\FSa(n)\,\}
$$
is canonical, and $A_{\FS{\shortrightarrow}\FSa}\FS=\FSa$.
\qed
\end{theorem}

If the replica has a complete log of the commands executed by the user (as
in an operation-based update detector), that is, if the updates have been
collected in a command sequence $\alpha_i$ which transforms $\FS$ to
$\FS_i$, then $\alpha_i$ can be transformed into the requested canonical set
as claimed by Theorem \ref{thm:basic}\,c). As detailed in Algorithm
\ref{alg:canonize} in Section \ref{sec:algorithms} this transformation can
be done in near linear time in the size of $\alpha_i$, which can be much
faster than traversing the whole filesystem.

% ------------------------------------
\subsection{Synchronization}\label{subsec:marger}

The central server, having received the canonical sets from the replicas
describing the local changes, must resolve all conflicts between the updates
and generate a common, synchronized filesystem. Conflict resolution should
be intuitively correct, thus discarding all changes made by the replicas is
not a viable alternative. While the majority of practical and theoretical
synchronizers do not present any rationale to explain their specific
conflict resolution approach \cite{syncpal-thesis}, two notable exceptions
\cite{TSR15}, \cite{NSC16} describe high-level consistency philosophies. In
\cite{TSR15} the main principles are \emph{no lost update} (preserve all
updates on all replicas because these updates are equally valid), and
\emph{no side effects} (do not allow objects to unexpectedly disappear).
While these principles make intuitive sense, neither can possibly be upheld
for every conflict. In \cite{NSC16} the relevant consistency requirements
are \emph{intention-confined effect} (operations applied to the replicas by
the synchronizer must be based on operations generated by the end-user), and
\emph{aggressive effect preservation} (the effect of compatible operations
should be preserved fully, and the effect of conflicting operations should
be preserved as much as possible). These requirements are, in fact,
variations of the OT consistency model \cite{SE98}. Note that the other two
OT principles -- convergence and causality preservation -- do not apply to
filesystem synchronizers.

In keeping with the prescriptive nature of the above principles, we proceed
by \emph{defining} what a synchronized state is, rather than creating it by
some ad hoc method. Suppose that the original filesystem is $\FS$, the
modified filesystem at the $i^\textrm{th}$ replica is $\FS_i=A_i\FS$, where
$A_i$ is the canonical set submitted as input to the reconciler. The
synchronized or \emph{merged state} $\FSa$ is determined by the canonical
set $M$ called \emph{merger} such that $\FSa=M\FS$ where $M$ satisfies the
following two conditions:
\begin{itemz}
\item[1)]
every command in $M$ is submitted by one of the replicas;
\item[2)]
the canonical set $M$ is maximal with respect to the first condition.
\end{itemz}
The first condition ensures that the synchronization satisfies the
\emph{intention-confined effect:} there are no surprise changes in the
merged filesystem. The \emph{aggressive effect preservation} is guaranteed
by the second condition. As $M$ is maximal it preserves as much of the
intention of the users as possible. In general, there can be many different
mergers satisfying these conditions. The reconciler must choose one of them
either automatically using some heuristics, or manually as instructed by the
user.

Observe that the canonical sets $A_i$ describing the local changes are
jointly refluent (see Section \ref{subsec:refluent}), as all of them can be
(were) applied to the original filesystem $\FS$. The formal definition of a
merger is as follows.

\begin{definition}[Merger]\label{def:merger}
The merger of the jointly refluent canonical sets $A_1, \dots,A_k$ is a
maximal canonical set $M\subseteq \bigcup_i A_i$. The corresponding
synchronized state of the replicas $A_1\FS,\dots,A_k\FS$ is $M\FS$.
\end{definition}

This definition requires the absolute minimum. Due to its simplicity it is
clear, intuitively appealing, and captures the desired properties of a
synchronized state. It remains to be seen whether it is also sufficient or
it is an oversimplification. % requiring additional refinement to be usable.
Surprisingly, this definition is indeed sufficient. Mergers provided
by Definition \ref{def:merger} satisfy many additional desirable properties
without any further requirements. Some of these properties are discussed
in the next subsections.

For the rest of this section we fix the canonical sets $A_i$ and the
original filesystem $\FS$ so that the current state of replica $i$ is the
valid filesystem $\FS_i=A_i\FS$. Consequently none of the sets $A_i$ breaks
$\FS$, and therefore the command sets $A_i$ are jointly refluent.

As an example, suppose $\FS$ contains directories at \textsf{/a},
\textsf{/a/b}, and single file at \textsf{/a/b/c}. The first replica deletes
the file and all directores above it. The second replica creates a file
below \textsf{/a/b}. The third replica creates the same file but with
different content, and also creates a copy of that file below \textsf{/a}.
The canonical sequences describing these changes are
\begin{itemz}[5pt plus 2pt]
\item[$A_1$]\hspace{-1em}${}=\{\sigma_1,\sigma_2,\sigma_3\}$, where $\sigma_1=\(\textsf{/a/b/c},\f_o,\E)$,
 $\sigma_2=\(\textsf{/a/b},\D,\E)$, $\sigma_3=\(\textsf{/a},\D,\E)$;
\item[$A_2$]\hspace{-1em}${}=\{\tau \}$, where
 $\tau= \(\textsf{/a/b/z},\E,\f_z)$;
\item[$A_3$]\hspace{-1em}${}=\{\rho_1,\rho_2\}$, where
 $\rho_1= \(\textsf{/a/z},\E,\f_u)$, and $\rho_2=\(\textsf{/a/b/z},\E,\f_u)$.
\end{itemz}
It is clear that $\sigma_3$ is in conflict with all commands in $A_2$ and
$A_3$; $\sigma_2$ is in conflict with $\tau$ and $\rho_2$. Commands $\tau$
and $\rho_2$ are also in conflict; and a merger containing $\sigma_3$ must
also contain $\sigma_2$. $\sigma_1$ is compatible with all other commands,
thus it will be in each possible mergers. If $\sigma_3$ is in the merger,
then it can contain no commands from either $A_2$ or $A_3$. If $\sigma_2$ is
not present, then only the conflict $\tau$ vs{.} $\rho_2$ remains. Thus
there are four mergers, namely
\begin{align*}
M_1&=\{\sigma_1,\sigma_2,\sigma_3\}, ~~~
M_2=\{\sigma_1,\sigma_2,\rho_1\},\\
M_3&=\{\sigma_1,\tau,\rho_1 \}, ~~\mbox{ and }~~
M_4=\{\sigma_1,\rho_1,\rho_2 \}.
\end{align*}
Each merger describes a possible synchronized state, and each one can be the
desired one under the right circumstances.

% ------------------------------------
\subsection{Mergers are applicable to the filesystem}\label{subsec:mergerOK}

While Definition \ref{def:merger} does not require $M$ to work on the
original filesystem $\FS$, it never breaks it as proved in Proposition
\ref{prop:merger} in Section \ref{sec:theory}. Every merger creates a
meaningful synchronized state and never breaks the original filesystem.

% ------------------------------------
\subsection{Mergers can be created in near linear time}\label{subsec:create-merger}

Mergers can be created by a simple greedy algorithm. Proposition
\ref{prop:extend} in Section \ref{sec:theory} states that a non-maximal
canonical subset of $\bigcup_i A_i$ can always be extended by some command
from $\bigcup_i A_i$ so that it remains canonical. Consequently, starting
from the empty set and adding commands from $\bigcup_i A_i$ one by one while
keeping the set canonical produces a merger. In particular, any canonical
subset of $\bigcup_i A_i$ can be extended to a merger. Since checking
whether a command set is canonical takes linear time (Algorithm
\ref{alg:canonical}), this na\"ive approach requires cubic time. Algorithm
\ref{alg:merger} creates a merger in near linear time. To generate all
mergers in nondeterministic linear time, we need a more sophisticated
algorithm discussed in Section \ref{sec:all-mergers}.

% ------------------------------------
\subsection{Mergers have an operational characterization}\label{subsec:operational}

The synchronized state defined by the merger $M$ as $\FSa=M\FS$ has a clear
operational characterization. The local replica $\FS_i$ can be transformed
into the merged state by first rolling back some of the local operations
executed on that replica, then applying additional commands executed on
other replicas. By Proposition \ref{claim:forward}, $M\cap A_i$ is an
initial segment of both $A_i$ and $M$, thus
$$
  A_i \equiv (M\cap A_i)\circ(A_i\setm M).
$$
Rolling back the commands in $A_i\setm M$, that is, executing the canonical
set $(A_i\setm M)^{-1}$ on $\FS_i$ gives $(M\cap A_i)\FS$. Then, applying
the canonical set $M\setm A_i$ yields the filesystem
$$
    (M\cap A_i)\circ(M\setm A_i)\FS=M\FS=\FSa.
$$
In summary, the $i$-th replica should execute the command set
$$
   (A_i\setm M)^{-1}\circ (M\setm A_i)
$$
on its local copy $\FS_i$ to transform it into the synchronized state
$\FSa$. The rolled back commands in $A_i\setm M$ give a clear indication of
the local changes that are discarded. This command set could be presented to
the user to decide whether some of them should be reintroduced.

\smallskip

Choosing the merger $M_3=\{\sigma_1,\tau,\rho_1\}$ in the example above, the
first replica should roll back $\sigma_2$ and $\sigma_3$ by executing
$\{\sigma_2^{-1},\sigma_3^{-1}\}$. These commands restore the directories at
\textsf{/a} and \textsf{/a/b}. They are followed by executing
$\{\tau,\rho_1\}$ which adds the two files.
To reach the same synchronized state the second replica need not roll back
any of its commands. Executing $\{\sigma_1,\rho_1\}$ directly deletes the
file at \textsf{/a/b/c} and creates the new file at \textsf{/a/z}.
Finally, the third replica should roll back $\rho_2$ (after this the
precondition in $\tau$ holds so $\tau$ will not break the filesystem when
executed), and then execute the commands $\tau$ and $\sigma_1$ in any order.

\smallskip

Observe that after synchronization none of the rolled back commands is
applicable anymore. $\sigma_2$ and $\sigma_3$ would delete directories which
are not empty, while $\rho_2$ would create a file which already exists. This
is true in general. The maximality of the merger set $M$ implies that none
of the rolled back commands can be executed directly on the synchronized
state $\FSa$. Either its input condition would fail (modification by some
other replica on that node took precedence), or the command would destroy
the tree property (deleting a non-empty directory, or creating a file under
a non-existent directory). Therefore the changes represented by the rolled
back commands can only be reintroduced in a different form.

% ------------------------------------
\subsection{Mergers can be created via conflict resolution}\label{subsec:merger-conflict}

Synchronizers typically work by identifying and resolving conflicts until no
conflicts remain. While Definition \ref{def:merger} specifies the
synchronized state directly, we can construct mergers via identifying
conflicts as well. Clearly two commands are in conflict if they cannot occur
together in the same canonical set. In our case, however, a weaker notion of
conflict also works.

\begin{definition}\label{def:conflict}
The commands $\sigma,\tau\in \bigcup_i A_i$ are in \emph{conflict} if either
\begin{itemz}[1pt]
\item[a)] they are different commands on the same node, or
\item[b)] the node of $\sigma$ is above the node of $\tau$,
$\sigma$ creates a non-directory and $\tau$ creates non-empty content.
\qedhere
\end{itemz}
\end{definition}

Commands $\tau= \(\textsf{/a/b/z},\E,\f_z)$ and
$\rho_2=\(\textsf{/a/b/z},\E,\f_u)$ from the example above are in conflict
as they are different commands on the same node. Command
$\sigma_3=\(\textsf{/a},\D,\E)$ creates a non-directory, thus it is in
conflict with every command below it which creates a file -- that is, all
commands in $A_2$ and $A_3$ --, but it is not in conflict with
$\sigma_1=\(\textsf{/a/b/c},\f_o,\E)$ and $\sigma_2=\(\textsf{/a/b},\D,\E)$
which create empty content..

By Proposition \ref{prop:conflict1} canonical sets (and thus mergers) do not
have conflicts at all, while Theorem \ref{prop:conflict2} claims that a
maximal command set of $\bigcup_i A_i$ without conflicts is a merger. The
synchronizer, having received the command sets $A_i$, can create a
\emph{conflict graph}, whose vertices are the commands in $\bigcup_i A_i$
and edges connect conflicting commands. Mergers correspond to the maximal
independent vertex sets of this graph. Creating a merger via conflict
resolution can therefore be done using the following procedure: pick an edge
of the graph representing a conflict between, say, $\sigma$ and $\tau$.
Choose either $\sigma$ or $\tau$ as the winner, and delete all vertices
connected to the winner (i.e., those vertices which cannot be in the same
merger as the winner). When there are no more edges, the remaining vertices
form a maximal independent set, a merger.

Using this conflict graph the synchronizer can make smart decisions, a
feature which is painfully missing in commercial and other theoretical
synchronizers. When choosing the conflict to be resolved and the winner of
the conflict, the decision can take into account not only local information
(the conflicting command pair), but also the effect of the decision on other
conflicts.

Creating and manipulating the conflict graph can be done in quadratic time
\cite{graph-alg}. As this graph has a very special structure and not too
many edges, the actual running time could be better. As our main interest in
this paper was developing subquadratic algorithms, we did not pursue this
line of research further. Using the conflict resolution strategy Algorithm
\ref{alg:merger} in Section \ref{sec:algorithms} creates a merger in near
linear time. Unfortunately, this algorithm, as explained later, cannot
generate all mergers in nondeterminstic linear time. For that task we need
further ideas explored in Section \ref{sec:all-mergers}.

% ------------------------------------
\subsection{Mergers support asynchronous and offline synchronization}\label{subsec:async-sync}

As usual, let the filesystem at the $i$-th replica be $\FS_i=A_i\FS$, with
the replica sending the canonical command set $A_i$ to the server for
synchronization. Section \ref{subsec:operational} discussed that when the
server returns the merger $M$, the replica should execute
$$
    (A_i\setminus M)^{-1} \circ (M\setminus A_i)
$$
on the local copy to transform it to the synchronized state $\FSa=M\FS$.
Optionally, the replica could present the conflicting command set
$A_i\setminus M$ to the user for inspection.

Suppose the replica $\FS_i$ has not been locked, and by the time the reply
$M$ arrives from the server, it has changed to $\FS'_i= A'_i\FS$, where
$A'_i$ is a new canonical set describing the differences between the current
state $\FS'_i$ and the original common state $\FS$; see Figure
\ref{fig:sync2}. In this case, the local machine should transform the
replica to the synchronized state $\FSa$ and carry over those extra changes
that are still executable. To this end it invokes the the synchronization
algorithm for the canonical sets $A'_i$ and $M$ making sure that the
returned merger $M^*$ contains $M$. (This can be achieved e.g., by
constructing the merger via conflict resolution and making sure that each
conflict is resolved in favor of a command in $M$.) By Claim
\ref{claim:forward} $M$ will be an initial segment of $M^*$, thus
$M^*=M\circ M'$ where $M'$ is the canonical set $M'=M^*\setm M$. Apply the
commands
$$
   C'_i =  (A'_i\setminus M^*)^{-1} \circ (M^*\setminus A'_i)
$$
to the filesystem $\FS'_i$ and return $A'_i\setminus M^*$ as the conflicting
command set. At this moment the local filesystem is
$$
    C'_i\FS'_i = C'_i(A'_i\FS) = M^*\FS = M'(M\FS)= M'\FSa,
$$
which is exactly the synchronized filesystem $\FSa$ to which the command set
$M'$ has been applied. The commands in $M'$ can be incorporated easily into
the next round of synchronization.

In fact, the same method works even if the replica did not take part in the
process of determining the merger $M$. Thus \emph{latecomers} or
\emph{offline}
replicas, who did not participate in determining the merged state, can still
upgrade to it without losing their ability to take part in subsequent
synchronization rounds.

% -----------------------------------------------------------------------
\section{Algorithms}\label{sec:algorithms}

Let us begin with some of the properties, assumptions and decisions the
algorithms rely on. In our model a filesystem command has three components:
the node it operates on and the input and output values. Each component is
stored in some constant space (using pointers, if necessary). Commands can
be sorted using any lexicographic order on the nodes which is consistent
with the ``parent'' function. In the standard example, a node (path) name is
a sequence of identifiers separated by the slash character. With the
assumption that comparing two path strings lexicographically takes constant
time, sorting $n$ filesystem commands can be done deterministically in
$O(n\log n)$ time \cite{knuth97}. We also assume that other path
manipulating algorithms, such as returning the parent of a node, or deciding
whether a node is above another one, also take constant time. Similarly, we
presuppose that operations on filesystem values (determining their type,
checking their equality or comparing them for sorting) can also be done in
constant time.

Almost all algorithms assume that the input commands are in a doubly linked
list sorted lexicographically by the nodes of the commands. The time and
space complexity estimates of the algorithms typically exclude this sorting
time.
A proof-of-concept implementation of the algorithms presented in this paper
in Python can be found at
\url{https://github.com/csirmaz/algebraic-reconciler}.

% ------------------------------------
\subsection{The $\up$ structure}

Our first algorithm will be used many times, frequently tacitly, as an
auxiliary tool. It enhances a set of nodes by adding an extra $\up$ pointer
between the nodes. This pointer at node $n$ is $\bot$ if no other node in
the set is above $n$; otherwise, it points to the node in the set which is
its lowest ancestor. In particular, if the parent of $n$ is also in the set,
then $\up(n)$ is the parent of $n$.

% CSequence.add_up_pointers() 
\begin{algorithm}[Adding $\up$ pointers, Code \ref{code:1}]\label{alg:add-up}
Sort the nodes lexicographically, and check them in increasing order.
Suppose we have finished processing node $n$, and the next node in the list
is $m$. Find the first node in the sequence $n$, $\up(n)$, $\up(\up(n))$,
etc{.} which is above $m$. If found, set $\up(m)$ to this node. If none of
them is above $m$ or $m$ is the first node, then set $\up(m)$ to $\bot$.
\end{algorithm}

%%%%%%%%%%%%%%%%%%%%%%%%%%%%%

\begin{pseudocode}{%
Given a lexicographically sorted sequence of commands, add the
$\up$ pointers.\label{code:1}}
\ForEach{command \textbf{in} sequence}
  \If{this is the first command}
   \State      command$\pt$up${}\gets \bot$
   \Else
     \State  upCommand $\gets$ previousCommand
   \Loop{}
       \If{upCommand${}=\bot$ \textbf{or} \\ \hbox{\quad} upCommand$\pt$node
              is an ancestor of command$\pt$node}
         \State    command$\pt$up $\gets$ upCommand
         \State  \textbf{exit loop}
        \Else
          \State  upCommand $\gets$ upCommand$\pt$up
      \EndIf
  \EndLoop
 \EndIf
\EndFor
\end{pseudocode}
%%%%%%%%%%%%%%%%%%%%%%%%%%%%%%%

For correctness, observe that in the namespace the sequence $n$, $\up(n)$,
$\up(\up(n))$, etc{.} defines the right boundary of the nodes processed up
to $n$. Since the next node $m$ is to the right of the earlier nodes, its
ancestors in the given set must be in this list. After sorting, the running
time is linear as each $\up$ link is compared and discarded at most once,
and each $\up$ link is filled exactly once.\footnote{We are grateful to
G\'abor Tardos for devising this algorithm. It is used here with his
permission.}

% ------------------------------------
\subsection{Checking and ordering canonical sets}

This algorithm checks whether the command set $A$ is canonical, assuming
that it is sorted lexicographically according to the nodes of the commands.
It checks the first two conditions of Definition \ref{def:canonical}
directly. Instead of the third condition ($A$ is $\orel$-connected), the
following clearly equivalent conditions are verified:
\begin{itemz}[1pt]
\item if $A$ contains a command on the node $n$ and also on an ancestor of
$n$, then it contains a command on the parent of $n$;
\item if $\sigma,\tau\in A$ are on parent--child nodes, then either
$\sigma\orel\tau$ or $\tau\orel\sigma$.
\end{itemz}

% CSequence.is_set_canonical()
\begin{algorithm}[Determining if $A$ is canonical, Code \ref{code:2}]\label{alg:canonical}

Start with all commands in $A$ arranged in a doubly linked list according to
a lexicographic order of the command nodes. Loop through the commands and
check that there is one command on each node, and none of the commands is a
null-command. Run Algorithm \ref{alg:add-up} to define the $\up$ pointers.
Loop through the commands. If the $\up$ pointer is not $\bot$, then it must
point to the parent node; moreover, the current command and the command at
the parent must be $\orel$-related.
\end{algorithm}

%%%%%%%%%%%%%%%%%%%%%%%%%%%%%%%%
\begin{pseudocode}{%
Check whether a set of commands is canonical.\label{code:2}}
% \State sequence $\gets$ commandSet to sequence
\LComment{sort {\upshape sequence} lexicographically by the nodes of the commands}%; in case of equivalence keep original order}
\LComment{ add $\up$ pointers using Algorithm \ref{alg:add-up}}
\ForEach{command \textbf{in} sequence}
  \If{previousCommand$\pt$node $=$ command$\pt$node}
       
       \State \textbf{return false}
       \Comment{not canonical as multiple commands are on the same node}
%        \Comment{Not canonical as multiple commands on same node}
 \EndIf
 \If{ command$\pt$up${}\neq\bot$ \textbf{and not} (\\
        \hbox{\quad}$\langle$command$\pt$up, command$\rangle$ is a constructor pair
        \textbf{or}\\
        \hbox{\quad}$\langle$command, command$\pt$up$\rangle$ is a destructor pair
    )}
       \State \textbf{return false}
       \Comment{\hfill not canonical because the closest command on an ancestor
                is not on the parent, or they do not form a valid pair}
  \EndIf
%   \State previousComand $\gets$ command
\EndFor
\State \textbf{return true} \Comment{this is a canoncal set}
\end{pseudocode}
%%%%%%%%%%%%%%%%%%%%%%%%%%%%%%%

A canonical set $A$ can always be ordered to honor $\orel$. Perhaps the
simplest way to obtain such an ordering is to make two passes through the
lexicographically sorted set $A$, as is done by Algorithm
\ref{alg:ordering}. 

% CSequence.order_set()
\begin{algorithm}[Ordering a canonical set, Code \ref{code:3}]\label{alg:ordering}

Sort commands in a canonical set lexicographically. First, scan the commands
forwards (top-down) extracting constructor commands, and place them at the
beginning of the output sequence. Second, place the remaining commands on
the output sequence in reverse lexicographical order (bottom-up). This
includes destructors and edit commands matching $\(\bullet,\F,\F)$. It is
clear that this sequence order honors the $\orel$ relation.
\end{algorithm}

%%%%%%%%%%%%%%%%%%%%%%%%%%%%%%%%
\begin{pseudocode}{%
Order a canonical command set and return a canonical sequence.
\label{code:3}}
\State sequence $\gets$ commandSet in some order
\LComment{sort {\upshape sequence} lexicographically by the nodes of the commands}%; in case of equivalence keep original order
%\State output $\gets [\,]$

\ForEach{command \textbf{in} sequence}
    \If{ command is a constructor}
     \State   \textbf{push} command \textbf{on} output
    \EndIf
\EndFor
\ForEach{ command \textbf{in} sequence \textbf{backwards}}
    \If{ command is not a constructor}
    \State  \textbf{push} command \textbf{on} output
    \EndIf
\EndFor
\State \textbf{return} output
\end{pseudocode}
%%%%%%%%%%%%%%%%%%%%%%%%%%%%%%%%%

Both Algorithm \ref{alg:canonical} and Algorithm \ref{alg:ordering} of 
this section clearly run in near linear time.

% ------------------------------------
\subsection{Transforming a sequence to a canonical set}

Given a non-breaking command sequence $\alpha$, Algorithm \ref{alg:canonize}
creates a canonical set $A$ which semantically extends $\alpha$ in near
linear running time.

% CSequence.get_canonical_set()
\begin{algorithm}[Command sequence to canonical set, Code \ref{code:4}]\label{alg:canonize}

Sort the commands in $\alpha$ in a lexicographic order by their nodes,
retaining the original order where they are on the same node. Process them
from left to right. For any consecutive sequence of commands that are on the
same node (including one-element sequences), define a replacement command
that has the input value of the first command and the output value from the
last. If the two values are different, add the replacement command to the
result set.
\end{algorithm}

If $\alpha$ may be breaking, it is easy to check that the output and input
values of neighboring commands on the same node are equal, or if the
resulting set is indeed canonical. Failure of these checks implies that
$\alpha$ was breaking, though the algorithm may also successfully convert a
breaking sequence to a non-breaking canonical set.

%%%%%%%%%%%%%%%%%%%%%%%%%%%%%%%%%
\begin{pseudocode}{%
Return the canonical command set that is the semantic extension of this sequence.
\label{code:4}}
\LComment{sort {\upshape sequence} lexicographically by the nodes of the commands;
    in case of equality keep the original order }
%\State output${}\gets\{\}$
\ForEach{command \textbf{in} sequence}
    \If{ this is the first command}
     \State   input $\gets$ command$\pt$inputValue
    \Else
        \If{ command$\pt$node $\neq$ prevCommand$\pt$node}
        \State  newCmd $\gets$ $\langle$prevCommand$\pt$node, input, prevCommand$\pt$outputValue$\rangle$
            \If{ newCmd is not a null command}
             \State %\textbf{push} newCmd \textbf{on} output
                    output $\gets$ output${}\cup \{ $newCmd$ \}$
            \EndIf
         \State  input $\gets$ command$\pt$inputValue
        \EndIf
    \EndIf
\EndFor
\If{ sequence was not empty}
   \State newCmd $\gets$ $\langle$lastCommand$\pt$node, input, lastCommand$\pt$outputValue$\rangle$
    \If{ newCmd is not a null command}
      \State %\textbf{push} newCmd \textbf{on} output
          output $\gets$ output${}\cup \{ $newCmd$ \}$
    \EndIf
\EndIf
\State \textbf{return} output
\end{pseudocode}
%%%%%%%%%%%%%%%%%%%%%%%%%%%%%%%%%%

% ------------------------------------
\subsection{Generating a merger in near linear time}\label{subsec:merger}

Theorem \ref{prop:conflict2} characterizes a merger of the jointly refluent
command sets $A_i$ as a maximal subset of $\bigcup_i A_i$ without conflicts.
This characterization can be turned into a greedy algorithm which generates
a merger in near linear time. Actually, the algorithm finds a maximal
independent vertex set of the conflict graph (discussed in Section
\ref{subsec:merger-conflict}) exploiting some special properties of this
graph.

According to Definition \ref{def:conflict}, commands $\sigma$ and $\tau$ are
in conflict if either they are acting on the same node; or if their nodes
are comparable, the upper command creates a non-directory and the lower
command creates a non-empty value. Loop through the commands in $\bigcup_i
A_i$ in a top-down order. At command $\sigma$, if $\sigma$ has been marked
as in conflict with some earlier command, then skip it. Otherwise keep
$\sigma$ and mark commands which are in conflict with $\sigma$ as
conflicting. It follows that if $\sigma$ is not skipped, it is not in
conflict with commands preceding it, and so conflicting commands are either
on the same node, or below the node of $\sigma$. To achieve the desired
speed, instead of scanning all subsequent commands immediately, we use lazy
bookkeeping. In essence, if $\sigma$ is selected, we flag its node to
remember to delete conflicting commands on descendant nodes. At each
subsequent node we check whether its parent has this flag. If yes, we flag
that node as well and process any conflicts accordingly. By Theorem
\ref{claim:check-refluent} the node set of jointly refluent canonical sets
is connected, thus this flag percolates properly to the descendants.

\begin{algorithm}[Generating a merger, Code \ref{code:5}]\label{alg:merger}
The inputs are the jointly refluent canonical sets $A_i$; the output is a
merger $M$. Sort the commands in $\bigcup_i A_i$ lexicographically and then
use Algorithm \ref{alg:add-up} to create the $\up$ pointers. Add a
``delete conflicts down'' flag to the node of each command, initially unset.

Loop through the commands of $\bigcup_i A_i$ in lexicographic order. At
command $\sigma$ at node $n$, check if the node of the command $\up$ points
to has the ``delete conflicts down'' flag set. If yes, then set this flag at
$n$ as well. If, additionally, $\sigma$ creates some non-empty content (it
is in conflict with a final command above it), then delete $\sigma$. If
$\sigma$ is not deleted, then mark it as ``final'' and delete all subsequent
commands on the same node $n$. If $\sigma$ has been marked ``final'' and it
creates a non-directory value, set the ``delete conflicts down'' flag at
$n$.

Commands marked as ``final'' form a maximal command set without conflicts,
thus they form a merger.
\end{algorithm}

%%%%%%%%%%%%%%%
\begin{pseudocode}{%
Given a set of jointly refluent canonical command sets, generate a merger.
\label{code:5}}
\State sequence $\gets$ union of commands in the command sets
\LComment{ sort {\upshape sequence} lexicographically by the nodes of the commands}
\LComment{ add $\up$ pointers using Algorithm \ref{alg:add-up}}
\ForEach{command \textbf{in} sequence}
 \If{ command$\pt$node $=$ deleteOnNode}
   \State \textbf{continue} \Comment{skip this command}
 \EndIf
 \If{ command$\pt$up${}\neq\bot$ \textbf{and}\\
        \hbox{\quad} command$\pt$up$\pt$node$\pt$delConflictsDown}
      \State command$\pt$node$\pt$delConflictsDown $\gets$ \textbf{true}
        \If{command$\pt$output${}\neq\E$} \Comment{non-empty}
           \State \textbf{continue} \Comment{skip this command}
        \EndIf
  \EndIf
    \State %\textbf{push} command \textbf{on} merger
        merger $\gets$ merger${}\cup\{$command$\}$
    \State deleteOnNode $\gets$ command$\pt$node
    \If{ command$\pt$output${}\neq\D$} \Comment{non-directory}
    \State
        command$\pt$node$\pt$delConflictsDown $\gets$ \textbf{true}
     \EndIf
\EndFor
\State \textbf{return} merger
\end{pseudocode}
%%%%%%%%%%%%%%%%%%%

Unfortunately this algorithm cannot generate all possible mergers. The only
non-deterministic choice it can make is picking the winner among commands on
the same node which are not in conflict with previous commands. (Algorithm
\ref{alg:merger} chooses the first such command.) Otherwise, when the
algorithm encounters a command for the first time, it puts it into the final
list even if there might be mergers which do not contain this command. The
more sophisticated Algorithm \ref{alg:xxx} generates all mergers in
nondeterministic near linear time.

The second step in asynchronous synchronization discussed in Section
\ref{subsec:async-sync} requires not only a merger, but a merger which
extends a given canonical subset $C$ of $\bigcup_i A_i$. With some tweaks
Algorithm \ref{alg:merger} can be used for this task as well. The idea is
that commands in $\bigcup_i A_i$ are scanned twice. First, all commands are
deleted which are in conflict with some command in $C$. Second, use the
remaining commands only and proceed as in Algorithm \ref{alg:merger}. The
first scan requires, however, not only a ``conflicts down'' flag, but also a
``conflicts up'' flag. To ensure that the algorithm spends linear time
handling the upward conflicts, it should check whether this flag is set
first, and if yes, quit the upward processing. Otherwise it should set the
flag, process the node, and continue processing at the parent node. We leave
it to the interested reader to work out the details.

% =========================================================================
\section{Theory}\label{sec:theory}

This section contains supporting theoretical results from the Algebraic
Theory of Filesystems. Some of the results have been used to justify the
correctness of algorithms presented in Section \ref{sec:algorithms}.
Algorithm \ref{alg:refluent} that checks whether some canonical sets are
refluent is presented in this section as it uses the specific
characterization given in Theorem \ref{claim:check-refluent}. First, we look
at conditions which guarantee that a canonical set is applicable to a
filesystem. Then, these conditions will be used to characterize refluent
canonical sets.

\begin{claim}\label{claim:1}
The canonical set $A$ is applicable to the filesystem $\FS$ if and only if
the following conditions hold for every command $\sigma=\(n,x,y)\in A$:
\begin{itemz}
\item[a)] $\FS(n)=x$;
\item[b)] if $\sigma$ is a destructor, then $\FS(n')=\E$ at every node $n'$ 
below $n$ not mentioned in $A$;
\item[c)] if $\sigma$ is a constructor, then $\FS(n')=\D$ at every node $n'$
above $n$ not mentioned in $A$.
\end{itemz}
\end{claim}
\begin{proof}
The conditions are necessary. Condition a) is clear. For b) and c) note that
no command in $A$ changes the filesystem value at $n'$, and after executing
$\sigma$, the value at $n'$ must be empty (or directory in case c),
respectively). To show that the conditions are sufficient, let $\sigma\in A$
for which there is no $\tau\in A$ where $\tau\orel\sigma$. Then $\sigma$ can
be executed on $\FS$ as $\FS(n)=x$ by condition a), and because if $\sigma$
is a constructor, then no commands on nodes above $n$ are in $A$, thus the
values at those nodes are $\D$; and if $\sigma$ is a destructor, no command
below $n$ is in $A$, thus all nodes there contain the empty value.
Furthermore, conditions a)--c) clearly inherit to the filesystem $\alpha\FS$
and the command set $A\setminus\alpha$.
\end{proof}

% --------------------------------
\subsection{Characterizing refluent sets}

\begin{claim}\label{claim:2}
The canonical sets $A$ and $B$ are refluent if and only if the following
conditions hold:
\begin{itemz}[1pt]
\item[a)] if $\sigma\in A$ and $\tau\in B$ are on the same node,
then their input values are the same;
\item[b)] if $\sigma,\tau\in A\cup B$ are on comparable nodes, then
for each node $n'$ between them there is a command in $A\cup B$ on $n'$;
\item[c)] suppose $\sigma,\tau\in A\cup B$ are on nodes $\parent n$ and $n$,
respectively. If one of the sets mentions $n$ but not $\parent n$, then the
input of $\sigma$ is $\D$; if one of the sets mentions $\parent n$ but not
$n$, then the input of $\tau$ is $\E$.
\end{itemz}
\end{claim}

\begin{proof}
The conditions are necessary. It is clear for a). For the other two
conditions suppose $A$ can be applied to $\FS$. Observe that according to
Claim \ref{claim:1} if $A$ has a command on node $n$ but not on nodes above
$n$, then on nodes above $n$ the filesystem must contain directories; and if
there are no commands in $A$ below $n$ then all nodes below $n$ must be
empty.

For the other direction we use Claim \ref{claim:1}, too. Set the content at
each node mentioned in $A\cup B$ to the common input value. Additionally,
set the content to directory at each node above a non-empty node and set the
content to empty below every non-directory node. Furthermore, for each
constructor command in $A\cup B$ on node $n$ set every node above $n$ not
mentioned in $A\cup B$ to a directory. Similarly, for each destructor
command in $A\cup B$ on node $m$ set every node below $m$ not mentioned in
$A\cup B$ to empty. Observe that values at nodes in $A\cup B$ did not change
due to b) and c). These assignments produce a valid filesystem which
satisfies the conditions of Claim \ref{claim:1}. \end{proof}

\begin{proposition}\label{prop:refluent}
If the canonical sets $A_i: i\le k$ are pairwise refluent, then they are
jointly refluent.
\end{proposition}
\begin{proof}
Mimicking the proof of Claim \ref{claim:2}, construct the filesystem $\FS$
as follows. Start with all empty nodes. For each command in $A_i$, set the
value at the node of the command to its input value. Each node gets the
same value as the sets $A_i$ are pairwise refluent. For the same reason, if
a node gets a non-empty value, then all nodes above it can be set to be a
directory. Next, if $\sigma\in A_i$ is a destructor and $n'$ below $n$ is
not mentioned in $A_i$ but $n'$ is not empty, then it is set by some $A_j$,
and then $A_i$ and $A_j$ are not refluent. Finally, if $\sigma\in A_i$ is a
constructor, $n'$ is above $n$ not mentioned in $A_i$, then $n'$ should be
set to a directory. If it cannot be done because either $\FS(n')$ has been
set to a different value, or some node above $n'$ has been set to a
non-directory, then again we get an $A_j$ such that $A_i$ and $A_j$ are not
refluent.
\end{proof}

\begin{claim}\label{claim:forward}
Suppose the canonical sets $A$ and $B$ are refluent. Then $A\cap B$ is an
initial segment of $A$.
In particular, $A \equiv (A\cap B)\circ (A\setminus B)$.
\end{claim}
\begin{proof}
Suppose $A\cap B$ is not empty and let $\sigma$ be one of the common
commands. By Claim \ref{claim:2}, if $\tau\in A$ and $\tau\orel\sigma$, then
$\tau$ must be in $B$ as well. Thus $A\cap B$ contains a command which is an
initial segment both in $A$ and in $B$. Delete this command from $A$ and $B$
and apply this Claim recursively to the remaining commands.
\end{proof}

Algorithm \ref{alg:refluent} below checks whether the collection $\{A_i:
i\le k\}$ of canonical sets are jointly refluent using the characterization
proved in Theorem \ref{claim:check-refluent}. For stating the theorem,
define, for any node $n\in\N$ and for $i\le k$, the index set $I_n$ as
$$
  I_n=\{i:{}\mbox{there is a command in $A_i$ on node }n\}.
$$
Assuming further that all commands in $\bigcup_i A_i$ on node $n$ have the
same input value, this common value is denoted by $x(n)$.

\begin{theorem}\label{claim:check-refluent}
The canonical sets $A_i$ are jointly refluent if and only if the following 
conditions hold: 
\begin{itemz}[1pt]
\item[a)] all commands on node $n$ have the same input value;
\item[b)] if $m$ is above $n$ and neither $I_n$ nor $I_m$ are empty, then
$I_{\parent n}$ is non-empty as well;
\item[c)] if $x(\parent n)\neq\D$, then $I_n\subseteq I_{\parent n}$;
\item[d)] if $x(n)\neq \E$, then $I_{\parent n}\subseteq I_n$.
\end{itemz}
\end{theorem}
\begin{proof}
Let us remark that condition b) is equivalent to requesting that if $n$ and
$m$ are comparable, neither $I_m$ nor $I_n$ are empty, then $I_{n'}$ is not
empty for nodes between $n$ and $m$.

To check that the conditions are necessary, let $\FS$ be a filesystem on
which all $A_i$ work. Then $\FS(n)=x(n)$ for all nodes mentioned in
$\bigcup_i A_i$, giving condition a). Condition b) follows from part b) of
Claim \ref{claim:2} applied to the refluent sets $A_i$ and $A_j$ where $i\in
I_n$ and $j\in I_m$. To check c), assume $x(\parent n)\neq D$. Then
$\FS(n)=\E$, and so if $A_i$ has a command on $n$ (that is, $i\in
I_n$), then $A_i$ changes $\FS(n)$ to a non-empty value, thus at the end the
filesystem must have a directory at $\parent n$. As $A_i$ does not break
$\FS$ and $\FS(\parent n)$ is not a directory, it must contain a command on
$\parent n$, thus $i\in I_{\parent n}$, as
required by c).
Similarly, if $x(n)\neq \E$ (in which case $\FS(\parent n)=\D$) and $A_i$
has a command on $\parent n$, (that is, $A_i$ sets $\FS(\parent n)$ to be a
non-directory), then $A_i$ must also set $\FS(n)$ to be empty, thus $i\in
I_n$, as required by d).

For the reverse implication it suffices to show that Claim \ref{claim:2} is
true for every pair $A_i$ and $A_j$, and then apply Proposition
\ref{prop:refluent}. Condition a) of Claim \ref{claim:2} is immediate from
a). For the rest we first remark that if $m$ is the parent of $n$ and none
of $I_m$ and $I_n$ are empty, then $\type x(m) \ge \type x(n)$. Indeed, if
$x(m)\neq\D$, then $I_n\subseteq I_m$ by c), thus there is a canonical $A_k$
which has commands on both $n$ and $m$, consequently we must have $x(n)=\E$.
Similarly, if $x(n)\neq\E$ then $I_m\subseteq I_n$, which implies similarly
that $x(m)=\D$.

Returning to checking conditions in Claim \ref{claim:2} for sets $A_i$ and
$A_j$, suppose $i\in I_m$ and $j\in I_n$ and $m$ is above $n$. Consider the
path between $m$ and $n$. By condition b) there are commands on every node
between $m$ and $n$, and by the previous paragraph the input types on these
nodes are non-increasing. If the next node below $m$ on the $m\,$---$\,n$
path is not $\E$, then d) gives that $A_i$ also has a command on that node,
too. Similarly, if the node immediately above $n$ is not $\D$, then by c)
$A_j$ has a command on that node. Consequently either there is a node
between $m$ and $n$ on which both $A_i$ and $A_j$ have a command, or
otherwise there is a command from $A_i$ and a command from $A_j$ on
parent--child nodes such that the former has input value $\D$ (as it is not
in $A_j$), and the latter has input value $\E$ (as it is not in $A_i$). In
all cases conditions b) and c) of Claim \ref{claim:2} hold, as required.
\end{proof}

Based on this characterization the following algorithm checks, in near
linear time, whether the canonical command sets $A_i$ for $i\le k$ are
refluent. The algorithm assumes that the sets $A_i$ are canonical.

\begin{algorithm}[Checking if canonical sets are refluent, Code \ref{code:6}]\label{alg:refluent}

Create a lexicographically sorted list of the commands in $\bigcup_i A_i$.
Using Algorithm \ref{alg:add-up} add $\up$ pointers, both for the commands
and nodes, meaning that the $\up$ pointer of command $\sigma\in A_i$ points
to the command in $A_i$ which is directly above $\sigma$ (if there is such a
command in $A_i$), while the $\up$ pointer at node $n$ points to the parent
of $n$ if there is any node above $n$ in the node set of $\bigcup_i A_i$.
Fill in the bitmaps $I_n$ stored at node $n$ according to which sets the
commands belong to.

All four conditions of Theorem \ref{claim:check-refluent} can be checked by
looping through the commands in lexicographic order. For condition b) note
that the precondition implies that the up pointer is filled in at $n$, and
it is enough to check that it points to $\parent n$.
Condition c) does not need to be checked where no up pointer points to
$\parent n$ as then all index sets below it are empty.
\end{algorithm}

The total processing time after
sorting is clearly linear if bitmap operations can be implemented in
constant time. Otherwise $I_n\subseteq I_{\parent n}$ can be checked in time
proportional to $|I_n|$ (which still gives a linear total time) by following
the command $\up$ links at node $n$. Checking $I_{\parent n}\subseteq I_n$
can be done by counting the number of command $\up$ links at node $n$ and
comparing it to the total number of elements in $I_{\parent n}$.

%%%%%%%%%%%%%%%%%%%%%%%%%%%%%%%%%
\begin{pseudocode}{%
Given a set of canonical command sets, determine if they are jointly
refluent.\label{code:6}}
\State sequence $\gets$ union of commands in the command sets
\LComment{ sort {\upshape sequence} lexicographically by the nodes of the commands}
\LComment{ add $\up$ pointers using Algorithm \ref{alg:add-up}}
\For{$i$ from $1$ to number of sets } \Comment{determine the sets $I_n$}
    \ForEach{ command \textbf{in} sets[$i$] }
      \State  command$\pt$node$\pt$index$\gets$%
              \hbox{command$\pt$node$\pt$index$\,\cup\{i\}$}
    \EndFor
\EndFor
\ForEach{ command \textbf{in} sequence }
   \If{ not the first command \textbf{and}\\
     \hbox{\quad}previousCommand$\pt$node $=$ command$\pt$node \textbf{and}\\ 
     \hbox{\quad}previousCommand$\pt$input $\neq$ command$\pt$input}
       \State \textbf{return false} \Comment{not refluent, condition a)}
   \EndIf
    \If{command$\pt$up${}\neq\bot$}
        \If{command$\pt$up$\pt$node is not the parent of\\
           \hbox{\quad} command$\pt$node }
\State      \textbf{return false} \Comment{not refluent, condition b)}
        \EndIf
      \If{command$\pt$up$\pt$input${}\neq\D$} \Comment{not a directory}
            \If{command$\pt$node$\pt$index is not a subset of\\
            \hbox{\quad} command$\pt$up$\pt$node$\pt$index}
\State        \textbf{return false} \Comment{not refluent, condition c)}
            \EndIf
        \EndIf
       \If{command$\pt$input${}\neq\E$} \Comment{not empty}
        \If{command$\pt$node$\pt$index is not a superset of\\
             \hbox{\quad} command$\pt$up$\pt$node$\pt$index}
\State       \textbf{return false} \Comment{not refluent, condition d)}
        \EndIf
    \EndIf
    \EndIf
\EndFor
\State \textbf{return true} \Comment{the command sets are refluent}
\end{pseudocode}
%%%%%%%%%%%%%%%%%%%%%%%%%%%%%%%%

% ------------------------------
\subsection{Mergers by conflict resolution}

This section presents a proof of the claim that the mergers are exactly the
maximal conflict-free subsets. Recall from Definition \ref{def:conflict}
that two different commands in $\bigcup_i A_i$ are in conflict if either
a) they are on the same node, or b) they are on comparable nodes, the node
on the higher node creates a non-directory and the command on the lower node
creates a non-empty content.

\begin{proposition}\label{prop:conflict1}
There are no conflicts in a canonical set.
\end{proposition}
\begin{proof}
A canonical set contains at most one command on each node, so assume
$\sigma$ and $\tau$ are on comparable nodes. Then there is an $\orel$-chain
between them (see Definition \ref{def:canonical}), thus either
$\sigma\orel\cdots\orel \tau$, or $\tau\orel\cdots\orel\sigma$. In both
cases either the command on the higher node creates a directory, or the
command on the lower node creates an empty content.
\end{proof}

\begin{theorem}\label{prop:conflict2}
Suppose the command sets $A_i$ are jointly refluent. $M\subseteq \bigcup_i A_i$
is a merger if and only if $M$ is maximal without conflicts.
\end{theorem}
\begin{proof}

By Proposition \ref{prop:conflict1} a merger does not contain conflicts,
thus it suffices to show that a maximal conflict-free set $M$ is canonical.
$M$ contains at most one command on each node by condition a) of Definition
\ref{def:conflict}. Let $\sigma$, $\tau\in M$ be on nodes $n$, $m$,
respectively such that $n$ is above $m$. Let moreover $\sigma\in A_i$ and
$\tau\in A_j$. We want to show that there is a $\orel$-chain in $M$ between
$\sigma$ and $\tau$. We know that $\sigma$ and $\tau$ are not in conflict.

Consider first the case when $\sigma$ creates a directory, that is, it
matches $\(n,\E\F,\D)$. As $A_i$, $A_j$ are refluent, let $\FS$ be any
filesystem on which both $A_i$ and $A_j$ work. Since $\FS(n)$ is not a
directory, all nodes in $\FS$ below $n$ are empty, in particular, $\tau$
matches $\(m,\E,\F\D)$. As the canonical $A_j$ does not break $\FS$, $A_j$
must contain commands on all nodes between $n$ and $m$, including $m$. If
$n$ is a parent of $m$, then $\tau\orel\sigma$, and we are done. If $n$ is
strictly above $m$, then we may assume that there are no commands in $M$ on
nodes between $n$ and $m$, thus no command on $\parent m$ either. But $A_j$
contains a command $\tau'=\(\parent m,\E,\D)$ (as $A_j$ does not break
$\FS)$, and $M\cup\{\tau'\}$ is conflict-free, contradicting the maximality
of $M$.

The second case is when $\tau$ creates an empty node, that is, it matches
$\(m,\D\F,\E)$. Similarly to the above, $\sigma$ matches $\(n,\D,\F\E)$, and
then either $\tau\orel\sigma$, or otherwise $A_i$ contains the command
$\tau'=\(\parent m,\D,\E)$ which can be added to $M$.
\end{proof}

\begin{proposition}\label{prop:merger}
Let $A_i$ be canonical sets, $M\subseteq \bigcup_i A_i$ be a merger. If
none of $A_i$ breaks $\FS$, then neither does $M$.
\end{proposition}
\begin{proof}

We use the conditions in Claim \ref{claim:1} to show that $M$ does not break
$\FS$. To this end let $\sigma=\(n,x,y)\in M$ so that $\sigma\in A_i$. Since
$A_i$ does not break $\FS$, condition a) follows. To check b) suppose
$\sigma$ is a destructor command, $n'$ is below $n$ and it is not mentioned
in $M$. If $n'$ is not mentioned in $A_i$ either, then condition b) holds as
$A_i$ is applicable to $\FS$. So suppose $\tau'\in A_i$ is on the node $n'$.
As $\tau'\notin M$, there must be a command $\tau\in M$ on node $m$ which is
in conflict with $\tau'$. Since $M$ has no command on $n'$ (but has a
command above $n'$), $m$ must be above $n'$. By Definition
\ref{def:conflict} $\tau'$ creates a non-empty content. Since
$\sigma$ and $\tau'$ are not in conflict (both are in the canonical set
$A_i$), $\sigma$ must create a directory. But this contradicts the
assumption that $\sigma$ is a destructor.

The case when $\sigma$ is a constructor and $n'$ is above $n$ is similar.
\end{proof}

\begin{proposition}\label{prop:extend}
Let $A_i$ be refluent canonical sets, $C\subseteq \bigcup_i A_i$ be
canonical. There is a merger extending $C$.
\end{proposition}
\begin{proof}
As commands in $C$ are not in conflict by Proposition \ref{prop:conflict1},
$C$ can be extended to be a maximal conflict-free subset of $\bigcup_i A_i$.
But this set is a merger by Theorem \ref{prop:conflict2}.
\end{proof}

%======================================================================
\section{Generating all mergers}\label{sec:all-mergers}

Algorithm \ref{alg:merger} in Section \ref{subsec:merger} cannot generate
all possible mergers in nondeterministic linear time. The modified algorithm
which creates a merger extending a given canonical subset $C$, however, can
be used for this purpose as follows. Pick a random subset $C$ of the
commands in $\bigcup_i A_i$ and check if $C$ is canonical using Algorithm
\ref{alg:canonical}. If yes, use the modified version of Algorithm
\ref{alg:merger} to create a merger extending $C$; otherwise use the
original version to create a merger.

While this algorithm clearly generates all mergers in nondeterministic
linear time, it is not satisfactory as it blindly guesses the final merger.
In this section we develop a more appealing approach by further exploiting
the structure of refluent canonical sets. Let us fix the jointly refluent
canonical sets $\{A_i:i\le k\}$, and consider all nodes mentioned in the
command set $\bigcup_i A_i$. Since the $A_i$ canonical sets are refluent, we
know that the input values of the commands on the same nodes are equal.

Observe that if there are any conflicts among the commands, then there is a
conflict of one or more of the following special types:
\begin{itemz}
\item[(1)] multiple different commands on the same node with a file input
value,

\item[(2)] a pair of commands matching $\(\parent n,\D,\E\F)$ and
$\(n,\E,\F\D)$,

\item[(3)] multiple different commands with an empty input value on the same
node,

\item[(4)] multiple different commands with a directory input value on the
same node.
\end{itemz}
We eliminate these conflicts in this order. First we consider conflicts of
type (1). They are necessarily on uncomparable nodes as in any filesystem
file nodes are on such nodes. Of the commands on the same node, we choose a
winner. If the winner is a destructor or an edit (matching $\(n,\F,\E\F)$),
we delete all commands below $n$ which create non-empty content. If the
winner is a constructor or an edit (matching $\(n,\F,\F\D)$), we delete all
destructor commands above $n$. Since these conflicts are on uncomparable
nodes, deletions triggered by one do not affect the conflicts on another.
Also, since all deleted commands are in conflict with the winner, an element
of the merger, we know that the merger will be maximal.

Next we consider conflicts of type (2). We mark all parent nodes with a
directory value that have a destructor command and which have a constructor
on an empty node on one of their children. We consider such parent nodes in
bottom-up order, and either keep the destructor command(s) on the parent, or
all the constructor commands on the children, without choosing a winner yet.
If the destructors are kept, we delete all commands below $\parent n$ which
create non-empty content. If the constructors are kept, we delete all
destructors on and above $\parent n$. Since the empty child nodes in these
conflicts are on uncomparable nodes, the deletions there are independent.
Deletions of commands creating non-empty content on other children can only
affect commands matching $\(n,\D,\F)$, which may be part of conflicts of
this type already resolved. However, since there is a destructor on $\parent
n$, there must be a $\(n,\D,\E)$ command on such children, so the resolution
of earlier conflicts are not affected by these deletions.

The deletions of destructors upwards are not independent, but as we proceed
in bottom-up order, they may resolve yet unresolved conflicts of type (2),
but will never interfere with conflicts already processed. We note that the
subsequent steps in the algorithm always choose a destructor or a
constructor command as the winner on a node if at this stage it has at least
one. This means that all commands deleted here are in conflict with a
command that will be part of the merger, ensuring its maximality.

Conflicts of type (3) are considered in a top-down order. We choose a single
winner command on each node. If the winner matches $\(n,\E,\F)$, then we
delete all constructor commands below $n$. As the deletions are downwards,
and we proceed top-down, they may resolve yet unresolved conflicts of type
(3), but will not interfere with winners already chosen. Also, deleted
commands are clearly in conflict with the winner.

Finally, conflicts of type (4) are processed in bottom-up order. We again
choose a single winner on each node. If it matches $\(n,\D,\F)$, then we
delete all destructors above $n$. It is again true that the deletions do not
affect winners already chosen, and that the maximality of the merger is
guaranteed.

Since we know that any conflict entails a conflict of one of the above
types, and as we have removed all such conflicts, we also know that the
merger constructed is not only maximal, but also conflict-free.

The algorithm sketched below realizes this idea, thus generates all mergers
in randomized near linear time. It assumes that the input command sets $A_i$
are canonical and refluent.
 
\begin{algorithm}[Generating all mergers]\label{alg:xxx}

Arrange the commands in $\bigcup_i A_i$ lexicographically and add the $\up$
pointers as in Algorithm \ref{alg:refluent}. Make several passes over the
commands dealing with conflicts (1)--(4) as indicated above. Each pass
handles commands either in top-down order of their nodes, or in the reverse
bottom-up order. Handling the command at node $n$ (which is the parent of
another node in case (2)) may result in deleting those commands at node $n$
which satisfy a certain property, deleting all commands \emph{above $n$}
which satisfy some other property, and deleting all commands \emph{below
$n$} satisfying a third property, or some combination of these
possibilities. The algorithm assumes that those deletions are performed
before proceeding to the next node.
 
In summary, make four passes through the commands, alternating top-down and
bottom up orders, and handle cases (1) to (4) in each pass. Return the set
of the final, non-deleted commands as the merger. The correctness and that
the algorithm can actually create all mergers by making appropriate choices
follow from the discussion above.
\end{algorithm}

The running time is linear if each pass can finish processing in linear
time. To ensure this, we keep additional flags at each node noting either
that required upward deletions have been done at and above this node, or
that downward deletions should be performed as necessary. Upward deletions
are performed immediately following the $\up$ pointers, but they abort once
encountering a node in which the relevant flag is already set. This ensures
that each node is visited at most once for this purpose, keeping the running
time linear.

Flags for downward deletions are checked whenever visiting a node in a
top-down pass. If the flag is set on the parent, set the flag on the current
node, and delete the necessary commands there. This ensures that the latest
deletions are applied just in time. During bottom-up passes, downward
deletions are not performed immediately as descendant nodes are not
processed again, but rather they are delayed until the next top-down pass or
an additional pass is executed for this purpose. As each flag is set at most
once, the running time is guaranteed to be linear. An easily accessible
implementation of this algorithm in Python can be found at
\url{https://github.com/csirmaz/algebraic-reconciler}.

%======================================================================

%%%%%%%%%%%%%%%
\begin{figure*}[!b!th]\centering
\begin{tikzpicture}
\draw[->] (0,-0.1)--(0,3.4);
\foreach \y in {1,2,3} {
\draw (0,\y) node[left] {\footnotesize \y}; 
\draw[ultra thin,gray] (-0.1,\y)--(11.0,\y);
}
\foreach \y in {0} {
\draw (0,2*\y) node[left] {\footnotesize \y}; 
\draw[->] (-0.1,2*\y)--(11.0,2*\y);
}

\draw (-0.8,1.6) node[rotate=90] {\footnotesize running time (sec)};

\foreach \x in {5,10,15,20,25,30}{
  \draw (\x*0.325,-0.1)--(\x*0.325,+0.1);
  \draw (\x*0.325,-0.07) node[below] {\footnotesize \x k};
}
\draw (0.1,-0.8)node[right] {\footnotesize total number of commands to
synchronize};
\definecolor{MycolorB}{rgb}{0.2357,0.3038,0.7918}
\definecolor{MycolorC}{rgb}{0.1739,0.3790,0.3613}
\definecolor{MycolorD}{rgb}{0.9361,0.4891,0.1668}
\definecolor{MycolorE}{rgb}{0.6992,0.8527,0.4888}
\definecolor{MycolorF}{rgb}{0.7338,0.7612,0.1635}
\definecolor{MycolorG}{rgb}{0.1966,0.7285,0.7622}
\definecolor{MycolorH}{rgb}{0.5825,0.8959,0.7113}
\definecolor{MycolorI}{rgb}{0.9618,0.2300,0.2120}
\definecolor{MycolorJ}{rgb}{0.5118,0.9632,0.9666}
\definecolor{MycolorK}{rgb}{0.2979,0.5561,0.2645}
\definecolor{MycolorL}{rgb}{0.2467,0.0027,0.3078}
\definecolor{MycolorM}{rgb}{0.0903,0.5002,0.1743}
\draw plot[only marks,mark=*,mark options={fill=MycolorB,fill opacity=0.4,opacity=0.4}] coordinates {
(0.39975,0.07573232675)
(0.28275,0.057928015225)
(1.42155,0.283680611)
(2.80995,0.656904507083333)
(0.51675,0.102655221725)
(0.66885,0.124463646275)
(0.24375,0.042903902425)
(0.11115,0.0198343472583333)
(0.16965,0.0285400840833333)
(1.24215,0.246289507725)
(2.80995,0.64577974525)
(0.12285,0.0207560465333333)
(2.59545,0.592601150833333)
(0.99645,0.199571153975)
(0.16965,0.03622351425)
(0.99645,0.189674790166667)
(0.55575,0.10889211875)
(0.15795,0.02917778405)
(0.28275,0.0488434064083333)
(2.59545,0.58638932975)
(2.12355,0.4583731885)
(1.56195,0.32944593825)
(0.07605,0.016709537245)
(1.84275,0.386130110583333)
(0.18135,0.034749487275)
(0.51675,0.0891325711666667)
(0.43875,0.0765479537333333)
(0.09945,0.021223391245)
(0.32175,0.05422384875)
(0.39975,0.0693859797583333)
(1.07835,0.207176246583333)
(0.43875,0.083719797525)
(0.12285,0.029334665025)
(0.75075,0.1514433475)
(0.11115,0.02148846448)
(0.91455,0.171461389583333)
(0.91455,0.1735571525)
(0.83265,0.167121243525)
(0.24375,0.049062090475)
(1.16025,0.2419465935)
(3.23895,0.7698938495)
(0.36075,0.0621698537416667)
(0.18135,0.0293871184083333)
(1.98315,0.42728831125)
(0.08775,0.0186800237575)
(0.32175,0.0586504480833333)
(1.56195,0.311054652916667)
(1.70235,0.354297071333333)
(1.07835,0.20849667875)
(3.02445,0.69380260875)
(0.14625,0.0263472275166667)
(0.75075,0.1405450695)
(0.13455,0.030940499225)
(0.09945,0.0159034934166667)
(0.66885,0.117228171083333)
(1.42155,0.27916255125)
(0.59475,0.1051395795)
(1.24215,0.241619070833333)
(1.84275,0.3848548225)
(0.08775,0.013971381925)
(0.07605,0.0127957702583333)
(0.55575,0.100737372416667)
(0.15795,0.0270596824166667)
(0.47775,0.091287771)
(0.59475,0.1128453155)
(3.02445,0.712199628333333)
(0.83265,0.15234277725)
(0.47775,0.0819241174)
(1.70235,0.35407329)
(0.14625,0.0288922557)
(0.13455,0.02090982575)
(1.98315,0.433724270333333)
(2.12355,0.475678694916667)
(1.16025,0.2238268405)
(0.36075,0.07069734)
(3.23895,0.7682508265)
};
\draw plot[only marks,mark=*,mark options={fill=MycolorC,fill opacity=0.4,opacity=0.4}] coordinates {
(0.217425,0.04714281325)
(2.530125,0.56576352925)
(0.989625,0.196398897)
(0.217425,0.0364395038166667)
(2.319525,0.4971811505)
(3.857425,0.9107787485)
(0.593125,0.101509701166667)
(0.147225,0.024242321675)
(2.530125,0.558628731666667)
(0.827125,0.148972780333333)
(4.500925,1.20646560666667)
(2.108925,0.443924005)
(0.199875,0.0324329555833333)
(2.108925,0.463569141083333)
(0.534625,0.106566686025)
(0.827125,0.162457040275)
(0.234975,0.049474863475)
(0.885625,0.154805938166667)
(2.319525,0.500542894)
(0.885625,0.170154965225)
(0.768625,0.147633985)
(0.252525,0.0500650915)
(4.179175,1.025429545)
(0.476125,0.091889825975)
(2.951325,0.661638837)
(0.164775,0.0286540278333333)
(0.112125,0.024028153755)
(1.726725,0.3502195775)
(0.359125,0.0603304423333333)
(3.857425,0.911191859166667)
(1.112475,0.218590375525)
(1.235325,0.233551502666667)
(1.603875,0.3289130405)
(0.534625,0.0925104545833333)
(0.234975,0.0392535466666667)
(0.129675,0.024100445475)
(1.481025,0.29768069275)
(0.651625,0.111444231)
(3.161925,0.723145841)
(1.726725,0.3475445845)
(0.651625,0.13216282)
(1.358175,0.28212261675)
(2.740725,0.61775388675)
(4.822675,1.29031582166667)
(2.740725,0.616670852833333)
(0.417625,0.08489414475)
(0.710125,0.124721446916667)
(0.182325,0.0307284413166667)
(4.822675,1.2623448215)
(2.951325,0.669616886333333)
(0.417625,0.0725427564166667)
(1.849575,0.373203920083333)
(0.710125,0.130883429975)
(1.481025,0.291593747833333)
(0.270075,0.0450960323333333)
(0.182325,0.035677428475)
(4.500925,1.16091918875)
(1.235325,0.2461238035)
(1.849575,0.3839332865)
(1.358175,0.260751908916667)
(0.112125,0.0178945085083333)
(0.270075,0.04761756875)
(1.603875,0.3207644545)
(0.476125,0.0799950429333333)
(0.164775,0.033118396525)
(1.112475,0.2093018005)
(0.593125,0.109889917525)
(0.129675,0.02781840305)
(4.179175,1.07806865925)
(0.989625,0.180985798833333)
(0.359125,0.07482583775)
(3.161925,0.725086621)
(0.199875,0.036615963775)
(0.252525,0.0397010565833333)
(0.768625,0.132575635333333)
(0.147225,0.0259947638)
};
\draw plot[only marks,mark=*,mark options={fill=MycolorD,fill opacity=0.4,opacity=0.4}] coordinates {
(0.8645,0.15631499725)
(2.2932,0.48169452225)
(2.457,0.52940633075)
(0.6305,0.119234568275)
(2.7963,0.612688233333333)
(0.195,0.03261352535)
(3.9195,0.91406108625)
(0.147225,0.031535733275)
(0.7085,0.120784952166667)
(1.3104,0.24877237725)
(0.2418,0.0400875224333333)
(3.3579,0.78219676575)
(0.3354,0.05949666425)
(4.2003,0.97543822075)
(0.195,0.0359794415)
(6.4064,1.75913182166667)
(5.9774,1.6443938775)
(3.0771,0.694735603333333)
(4.2003,0.987943895)
(1.1765,0.218275337525)
(0.147225,0.0256190454916667)
(0.5525,0.0955306224166667)
(3.6387,0.82616025625)
(2.2932,0.478231831)
(0.7865,0.14921226675)
(5.1194,1.41367019583333)
(0.3588,0.065383259)
(1.9656,0.3980762355)
(0.312,0.0534012283583333)
(3.0771,0.6858366095)
(3.9195,0.9070792075)
(0.2184,0.034411966675)
(6.4064,1.72776309)
(0.5525,0.102014952275)
(0.2418,0.047069603275)
(0.6305,0.107094201083333)
(2.7963,0.5952525645)
(0.1716,0.0281929219083333)
(1.0985,0.20204869675)
(2.1294,0.44641492375)
(2.1294,0.444086825333333)
(1.4742,0.280351419916667)
(0.469625,0.0822024126833333)
(1.0205,0.18229428025)
(2.457,0.519720804583333)
(0.9425,0.1868660715)
(1.0205,0.200247892225)
(0.2652,0.0433390735)
(1.638,0.3234450485)
(3.6387,0.834842184166667)
(0.7085,0.125303189775)
(0.3354,0.0516684226583333)
(0.2886,0.049473978525)
(3.3579,0.765141856416667)
(0.1716,0.0314173235)
(0.7865,0.136790876166667)
(0.9425,0.166670911833333)
(5.9774,1.59497429675)
(0.312,0.0548497575)
(1.0985,0.2085461165)
(5.1194,1.36640112725)
(0.2652,0.0402433428333333)
(1.9656,0.397757323666667)
(0.469625,0.087901940975)
(1.1765,0.211968431083333)
(1.8018,0.357581466916667)
(1.638,0.334292084)
(0.3588,0.0563362491666667)
(1.8018,0.3676856885)
(5.5484,1.54116724083333)
(0.2184,0.04210077375)
(5.5484,1.4997040305)
(0.2886,0.0475560020916667)
(1.4742,0.289977612)
(0.8645,0.1646130475)
(1.3104,0.247867485583333)
};
\draw plot[only marks,mark=*,mark options={fill=MycolorE,fill opacity=0.4,opacity=0.4}] coordinates {
(0.418275,0.07096711875)
(0.784875,0.147358067225)
(0.272025,0.04371439535)
(1.174875,0.211986311916667)
(6.381375,1.73661390225)
(2.450175,0.520684826)
(1.835925,0.37044406625)
(4.536675,1.14717049625)
(3.064425,0.67173229175)
(2.654925,0.567358128416667)
(3.064425,0.67286947575)
(7.453875,2.07434740916667)
(2.859675,0.619839917833333)
(1.835925,0.360088416166667)
(2.245425,0.458547346)
(6.381375,1.74733313333333)
(0.979875,0.185206304)
(1.272375,0.242188322725)
(0.301275,0.058977112475)
(2.040675,0.421771344666667)
(0.784875,0.136492829416667)
(7.453875,2.034547059)
(0.389025,0.063876455925)
(0.330525,0.05739600475)
(4.185675,0.977965836)
(0.882375,0.161057433916667)
(0.6825,0.128672270725)
(3.483675,0.789862879083333)
(6.917625,1.90727883833333)
(1.467375,0.289090074)
(1.272375,0.231746542416667)
(0.242775,0.042525284275)
(4.536675,1.08788414333333)
(2.040675,0.4244929915)
(2.245425,0.46617786625)
(2.859675,0.60790942925)
(0.272025,0.0558310905)
(3.834675,0.875248606666667)
(1.077375,0.200380023275)
(0.242775,0.0401496251666667)
(0.21255,0.03602294575)
(4.185675,0.98403153)
(6.917625,1.8708110895)
(2.450175,0.51193815575)
(1.631175,0.3213138535)
(2.654925,0.5542539015)
(0.6825,0.117123510166667)
(7.990125,2.23313019083333)
(3.834675,0.888435594)
(1.174875,0.22337427375)
(4.887675,1.2409866615)
(0.359775,0.0577784990833333)
(1.369875,0.252775091166667)
(0.979875,0.182175339)
(5.238675,1.36561498175)
(0.21255,0.04541418325)
(0.447525,0.0734764728583333)
(3.483675,0.79926723325)
(1.467375,0.271644422416667)
(5.238675,1.39423078583333)
(0.359775,0.066858534475)
(1.077375,0.194923031083333)
(0.418275,0.070849660025)
(0.330525,0.0535001598416667)
(0.389025,0.07048482725)
(1.369875,0.25770841825)
(0.447525,0.08692165575)
(7.990125,2.1359348585)
(4.887675,1.26144670333333)
(0.882375,0.176701045525)
(0.301275,0.0501259055083333)
(1.631175,0.3303060805)
};
\draw plot[only marks,mark=*,mark options={fill=MycolorF,fill opacity=0.4,opacity=0.4}] coordinates {
(1.29025,0.23425741325)
(1.52425,0.304347355)
(1.52425,0.284383413666667)
(2.19765,0.459460766666667)
(0.934375,0.17983832575)
(1.40725,0.26049118225)
(1.05625,0.189702189166667)
(1.75825,0.34434568575)
(1.40725,0.27998083025)
(0.289575,0.05576489625)
(0.32565,0.06559125125)
(0.36075,0.0641725385)
(8.93035,2.52666339)
(1.64125,0.31258076775)
(0.46605,0.0804871828416667)
(2.93475,0.635498780166667)
(7.64335,2.0589566845)
(1.29025,0.248257683775)
(0.50115,0.0828152268416667)
(2.68905,0.576079234916667)
(5.85585,1.567093485)
(5.01345,1.3245373975)
(3.18045,0.691821863583333)
(0.43095,0.0766935045)
(2.93475,0.64453852525)
(8.28685,2.21926220175)
(3.18045,0.698321082)
(2.44335,0.516403374333333)
(0.32565,0.0543339144166667)
(9.57385,2.84755172)
(0.289575,0.05270495085)
(0.43095,0.0689084355)
(6.27705,1.61907884725)
(8.93035,2.5235782175)
(1.64125,0.313014851416667)
(5.01345,1.2888229)
(1.17325,0.222454160725)
(0.50115,0.0978499555)
(2.19765,0.4550064275)
(0.39585,0.0655324050666667)
(5.43465,1.43451408583333)
(9.57385,2.73339653)
(0.39585,0.07174716875)
(4.17105,0.961631114166667)
(8.28685,2.29256965583333)
(1.05625,0.1995196595)
(6.27705,1.6872147625)
(2.68905,0.57308048475)
(3.42615,0.7613850765)
(3.67185,0.829016496166667)
(1.75825,0.337155244166667)
(3.67185,0.81991787725)
(2.44335,0.51321256525)
(0.934375,0.165612275166667)
(1.940575,0.4006152825)
(4.59225,1.16532185)
(1.17325,0.21278841375)
(1.940575,0.386279504833333)
(5.43465,1.44292192125)
(0.53625,0.0898539569166667)
(0.36075,0.0576372207583333)
(7.64335,2.1056618025)
(4.17105,0.960585771)
(0.53625,0.102396606475)
(0.46605,0.090254106975)
(4.59225,1.15362362)
(5.85585,1.53685537875)
(3.42615,0.760269652666667)
};
\draw plot[only marks,mark=*,mark options={fill=MycolorG,fill opacity=0.4,opacity=0.4}] coordinates {
(0.461175,0.084395640475)
(0.420225,0.07672725475)
(9.656075,2.8906160125)
(1.366625,0.268074989)
(3.419325,0.759723722916667)
(1.776125,0.3502288175)
(3.132675,0.67812681525)
(6.332625,1.69830077333333)
(9.656075,2.83057372)
(2.846025,0.61010479)
(2.846025,0.5987650545)
(0.543075,0.101145758775)
(2.548,0.526334184083333)
(0.420225,0.0684987408166667)
(5.349825,1.37547000375)
(0.3783,0.074595501475)
(0.461175,0.0754404405083333)
(1.776125,0.343722281166667)
(5.349825,1.4035862075)
(2.049125,0.39697061225)
(0.502125,0.0811092212833333)
(1.503125,0.279862413)
(3.419325,0.7528076135)
(0.624975,0.115101011)
(6.824025,1.80783686)
(7.315425,1.95568780775)
(0.584025,0.0972886264166667)
(0.502125,0.092717465475)
(10.406825,3.15670403583333)
(1.22525,0.238694847225)
(2.049125,0.40464249225)
(1.366625,0.251781736416667)
(10.406825,3.0934097525)
(1.639625,0.310584711416667)
(3.992625,0.916001786)
(1.22525,0.223635732583333)
(4.279275,0.976110515)
(0.543075,0.0874781779166667)
(6.824025,1.76457235975)
(3.132675,0.68631838075)
(8.905325,2.437398415)
(0.3783,0.0620494589166667)
(8.905325,2.438678055)
(1.912625,0.370939075583333)
(1.503125,0.28845296375)
(1.912625,0.38672036225)
(3.705975,0.823246265916667)
(4.279275,0.977080104)
(2.548,0.52833691475)
(5.841225,1.5582851175)
(0.624975,0.101980884833333)
(5.841225,1.52715812825)
(1.639625,0.327762098)
(3.705975,0.83765267025)
(4.858425,1.2446206415)
(6.332625,1.64768082075)
(7.315425,1.9847471675)
(0.584025,0.103800546025)
(4.858425,1.21255818416667)
(3.992625,0.894890469166667)
};
\draw plot[only marks,mark=*,mark options={fill=MycolorH,fill opacity=0.4,opacity=0.4}] coordinates {
(0.7137,0.117844930416667)
(0.5733,0.0937186195833333)
(10.1673,3.0384368575)
(6.1074,1.5576763125)
(0.6669,0.125092657475)
(10.1673,2.98149372)
(1.716,0.3430069345)
(0.478725,0.0768331261083333)
(7.7922,2.078377206)
(4.2315,0.962603501666667)
(6.1074,1.5928196915)
(4.5591,1.0404819225)
(4.2315,0.95703741625)
(1.872,0.365450604416667)
(3.237325,0.70955055875)
(1.716,0.32528447425)
(3.5763,0.7867909095)
(4.8867,1.20502825833333)
(8.3538,2.29164642916667)
(5.525325,1.38561958575)
(2.028,0.410788798)
(0.6669,0.112491993916667)
(2.34,0.47720345725)
(4.5591,1.037794923)
(0.5265,0.0874082449166667)
(0.478725,0.083572521275)
(7.7922,2.11660866583333)
(3.237325,0.692676326916667)
(2.184,0.44270623225)
(3.9039,0.8809293)
(0.7137,0.13008182)
(1.555125,0.289277479)
(3.9039,0.87777263775)
(7.2306,1.93503850166667)
(0.6201,0.102320195333333)
(5.525325,1.41401931916667)
(0.6201,0.124606332025)
(2.028,0.396611601166667)
(0.5733,0.10789428925)
(4.8867,1.198524407)
(3.5763,0.777806361583333)
(7.2306,1.88749878525)
(8.3538,2.2773293465)
(6.669,1.724927632)
(1.555125,0.2989918475)
(2.184,0.43740785275)
(0.5265,0.0960321025)
(2.34,0.47657820125)
(1.872,0.3635450645)
(6.669,1.769584085)
};
\draw plot[only marks,mark=*,mark options={fill=MycolorI,fill opacity=0.4,opacity=0.4}] coordinates {
(2.104375,0.42347570225)
(8.760375,2.34874099775)
(0.802425,0.141662357975)
(0.644475,0.114593986775)
(6.8445,1.77022734125)
(0.697125,0.12719061525)
(0.697125,0.1182466805)
(2.279875,0.4560415045)
(0.749775,0.123664575416667)
(7.496775,2.03882217666667)
(0.802425,0.13261122375)
(6.8445,1.777464045)
(4.00855,0.886408866666667)
(2.455375,0.504041412833333)
(4.757025,1.18419287166667)
(2.630875,0.546657258333333)
(9.392175,2.71866791583333)
(8.128575,2.19588545775)
(0.644475,0.106721847583333)
(5.494125,1.39852675166667)
(0.749775,0.14049498625)
(2.630875,0.5575019105)
(0.59085,0.09708878875)
(5.125575,1.290004115)
(0.59085,0.114125234)
(2.279875,0.461156543583333)
(4.00855,0.89288888575)
(4.388475,0.999951981666667)
(8.760375,2.376856115)
(4.757025,1.18858780025)
(9.392175,2.6698567425)
(1.924,0.3685028035)
(4.388475,1.0060128735)
(1.924,0.38801342875)
(7.496775,1.99996943325)
(5.125575,1.30733884166667)
(8.128575,2.19435151833333)
(5.494125,1.4127632675)
(2.104375,0.42289685925)
(2.455375,0.50543383375)
};
\draw plot[only marks,mark=*,mark options={fill=MycolorJ,fill opacity=0.4,opacity=0.4}] coordinates {
(8.304075,2.20973033166667)
(2.72675,0.564936866833333)
(0.89115,0.160065951525)
(0.83265,0.15117684675)
(9.02655,2.43362176416667)
(5.28255,1.34630340583333)
(2.331875,0.468830063083333)
(10.43055,3.0647487)
(9.02655,2.415401753)
(4.861675,1.19525507416667)
(8.304075,2.19659534125)
(2.331875,0.4787537185)
(9.72855,2.7688240675)
(0.83265,0.13997997)
(0.89115,0.153261930666667)
(5.69205,1.45729239725)
(0.77415,0.13317950225)
(9.72855,2.82287267166667)
(5.28255,1.36460965125)
(2.53175,0.528686935)
(0.714675,0.118797336416667)
(0.714675,0.1389149105)
(10.43055,3.06894953083333)
(6.10155,1.57186669833333)
(0.77415,0.14736304575)
(5.69205,1.460829845)
(6.10155,1.574418065)
(2.53175,0.520931334666667)
(4.861675,1.19084438075)
(2.72675,0.56544113075)
(2.92175,0.61638935775)
(2.92175,0.612311874833333)
};
\draw plot[only marks,mark=*,mark options={fill=MycolorK,fill opacity=0.4,opacity=0.4}] coordinates {
(0.8502,0.15846765145)
(9.90405,2.808819535)
(6.708975,1.743112378)
(9.90405,2.8915483375)
(3.212625,0.68351839025)
(5.7967,1.47260330833333)
(3.212625,0.6848109415)
(6.258525,1.61827853666667)
(10.696725,3.14551566666667)
(0.979875,0.174938659025)
(10.696725,3.15818831)
(0.915525,0.15539575175)
(0.915525,0.174865505025)
(2.77875,0.5882246865)
(5.7967,1.50587780725)
(0.8502,0.143715558)
(6.258525,1.6328321345)
(2.77875,0.572598625166667)
(0.979875,0.169332341916667)
(6.708975,1.72073509833333)
(2.998125,0.624702093166667)
(2.998125,0.6506987145)
};
\draw plot[only marks,mark=*,mark options={fill=MycolorL,fill opacity=0.4,opacity=0.4}] coordinates {
(6.813625,1.73074069333333)
(0.997425,0.189070724725)
(7.3164,1.902600887)
(3.5035,0.763433985)
(7.3164,1.9203400125)
(3.264625,0.691127327583333)
(3.5035,0.754278283583333)
(0.997425,0.1735111035)
(1.0686,0.189536858083333)
(1.0686,0.2027947085)
(3.264625,0.71048551125)
(6.813625,1.7505284545)
};
\draw plot[only marks,mark=*,mark options={fill=MycolorM,fill opacity=0.4,opacity=0.4}] coordinates {
(7.91245,2.07452007625)
(7.91245,2.05006716166667)
(1.15635,0.2038685685)
(3.7895,0.84239777075)
(3.7895,0.825199096583333)
(1.15635,0.218940525025)
};
\end{tikzpicture}
\caption{\fontsize{9}{10}\selectfont
Running time for generating the synchronized state on several synthetic data
sets. The running time depends only on the total number of filesystem
commands ($x$ axis), and not on the number of replicas (color).
}\label{fig:running-time}%
\end{figure*}
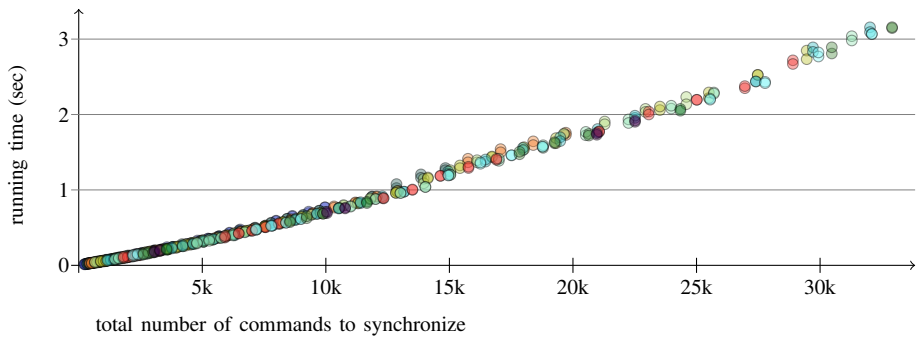
%%%%%%%%%%%%%%%

\section{Empirical results}\label{sec:results}

The performance of the algorithms has been tested on several synthetic data
sets. These sets consist of the collection of the canonical command
sequences the replicas executed on a common filesystem. This initial
filesystem is determined by two integer parameters $S$ and $T$. It has
non-empty nodes on the topmost three levels only. These nodes are labeled by
the paths \textsf{/i}, \textsf{/i/j} and \textsf{/i/j/k} where \textsf{i},
\textsf{j} and \textsf{k} are numbers between $0$ and $S-1$ such that
$(\textsf{i}, \textsf{j})$ and $(\textsf{j}, \textsf{k})$ are not farther from
each other modulo $S$ than $T$. The filesystem contains different files at
the non-empty nodes \text{/i/j/k}, and contains directories at all other
non-empty nodes. Typical parameter values are $T=2$ and $S=10$.

The set of command sequences to be synchronized also depends on the number
of users (replicas), which varies between $2$ and $S-1$. User \textsf{u} for
$0\le\textsf{u}\le S-1$ makes the following extensive changes on the
filesystem:
\begin{itemz}
\item[1)] deletes all existing files at \textsf{/i/u/k} for all \textsf{i} and
\textsf{k};
\item[2)] removes (the now empty) directories at \textsf{/i/u} for all
\textsf{i};
\item[3)] for each \textsf{x} in $(\textsf{u}-1, \textsf{u}, \textsf{u}+1)$
modulo $S$ and for all \textsf{i} and for all $\textsf{j}\neq\textsf{u}$
changes the file at \textsf{/i/j/x} (if exists) to a directory; finally
\item[4)] under each newly created directory creates $S$ new files with
unique content. These files are placed at the nodes with paths
\textsf{/i/j/x/l} where $0\le\textsf{l}<S$.
\end{itemz}
Depending on the parameter values $S$ and $T$ the number of necessary
filesystem commands achieving these changes varies between 100 and 8000 per
user. The instructions have been chosen so that there are both a large
number of conflicts and also a large number of non-conflicting command
pairs, forcing any conflict-based synchronizer to spend quadratic time even
to check the existence of conflicts. The command sequences have many
symmetries to ensure that the order in which the commands or command pairs
are processed has little or no effect on the running time.

In the experiments we have tried 20 different filesystems with the parameter
$S$ running from 5 to 14 (inclusive), $T$ running from 1 to $[(S-1)/2]$, and
the number of users running from 2 to $S-1$, inclusive. The general
synchronization Algorithm \ref{alg:xxx} from Section \ref{sec:all-mergers}
was called on the resulting collection of sequences to generate the first
and the first three possible synchronized states. Each run has been repeated
10 times to smooth out the effects of other programs running on the same
server. Figure \ref{fig:running-time} depicts the average time used to
generate a single synchronizing command set. The input size on the $x$ axis
is the total number of different commands in the sequences to be
synchronized. The average running time on the $y$ axis is in seconds. The
unoptimized Python program was running on a desktop machine with an Intel(R)
Core(TM) i5-8250U CPU @ 1.60GHz processor and 8G memory. The different
colors represent the number of replicas. The results confirm that the
running time is subquadratic in the total input size and does not depend on
the number of replicas.

%======================================================================

\section{Conclusion}\label{sec:conclusion}

This paper presented a provably correct synchronization algorithm running in
subquadratic time which can synchronize an arbitrary number of replicas. The
existence of such an algorithm was a long-standing open problem
\cite{syncpal-thesis} in the fields of Operation Transformation (OT)
\cite{SLL11} and Conflict-free Replicated Data Types (CRDT)
\cite{CRDT-overview}. Our work is based on the Algebraic Theory of
Filesystems (ATF) \cite{CC22a} which, in many respects, resembles both OT
and CRDT. Instead of traditional filesystem commands, ATF uses operations
enriched with contextual information similarly to OT and CRDT. It also
favors commutativity, but instead of requesting all operations to be
commutative, their non-commutative part can be isolated systematically and
handled separately. As a consequence, similarly to OT and CRDT, ATF can deal
with command sets instead of sequences where the execution order is not
specified, but the semantics of executing the commands \emph{in some order}
is defined unambiguously.

The underlying filesystem model, while arguably simplistic, retains the most
important high level and platform independent properties of real-life
filesystems; see Section \ref{subsec:filesystem} for a more detailed
discussion. The two most prominent omissions of this model are the lack of
\emph{directory attributes} (the model handles all directories as equal),
and \emph{links} breaking the regular tree-like structure of filesystem
paths. Both of these shortcomings, and how to circumvent them, are discussed
below.

Filesystem synchronization starts with \emph{update detection} run locally,
which extracts information encoding the state of the modified replica as
discussed in Section \ref{subsec:update-detector}. In the ATF framework this
information is a \emph{canonical command set} describing how to construct
the replica from the original filesystem.
 
The task of the synchronizer is to create a common, merged filesystem after
considering what changes have been applied to the replicas. In the ATF
framework this task amounts to creating another canonical command set, the
\emph{merger}, which transforms the original filesystem into the merged
filesystem. What this merger can be is described by two simple and
intuitively appealing principles: the \emph{intention-confined effect:}
operations in the merger should come from those supplied by the replicas,
and \emph{aggressive effect preservation}: the merger should contain as much
of those commands as possible. Definition \ref{def:merger} formalizes this
idea, and \emph{defines} what a synchronized state is. Sections
\ref{subsec:mergerOK} and \ref{subsec:operational} discuss that this
goal-driven definition automatically implies many operational properties,
such that a merger command set always defines a valid synchronized state;
and synchronization can be achieved from the local copy by rolling back some
of the local commands and executing additional ones originating from other
replicas. Section \ref{subsec:async-sync} shows that with minimal effort the
synchronization paradigm can be extended to tolerate replicas which do not
lock their filesystem -- allowing for asynchronous or optimistic
synchronization \cite{CRDT-overview}. Even replicas missing a
synchronization cycle can later upgrade to the synchronized state.

Algorithms described in Section \ref{sec:algorithms} give a high-level
description of a proof of concept implementation at
\url{https://github.com/csirmaz/algebraic-reconciler}. Algorithm
\ref{alg:merger} creates a merger in linear time after sorting the canonical
sets sent by the replicas. While this algorithm can make some
nondeterministic decisions, it cannot generate all possible mergers. Two
nondeterministic algorithms which can do so are sketched in Section
\ref{sec:all-mergers}; both of them run in near linear time. The first
algorithm uses the fact that Algorithm \ref{alg:merger} can recognize
mergers. It first creates a random subset of the supplied commands blindly,
and then checks if it is a merger. The second one, described as Algorithm
\ref{alg:xxx}, is more elaborate. It exploits the structure of refluent
command sets, and can assist in the synchronization process by highlighting
the consequences of different conflict resolutions.

\subsection{Node attributes}

The important task of consolidating different versions of the same document
(file value) was not considered as processing the internal structure of file
contents is outside the scope of filesystem synchronization. Files different
in content are considered to be different, and the synchronization algorithm
forces choosing one or the other. There is, however, an easy way of
incorporating third-party content-merging applications. This can be done by
pretending that there is only one possible file content, and using the ATF
framework to synchronize the \emph{structure} of the filesystems. When it
becomes clear which nodes contain file values, check for commands which
modified the actual content there, and use the external application to
determine the final file content. A similar approach can handle node
attributes by considering the changes made by the replicas at some node and
consolidating them. This approach, however, should be followed carefully. To
illustrate the problem, consider a directory at node $n$, which originally
had the ``private'' attribute. Replica $A$ changes this attribute to
``public'', while replica $B$, under the impression that the directory is
private, creates a file under it. When merging the attributes the change at
node $n$ is carried over, making the directory publicly available. This,
however, is clearly unacceptable. It is an interesting open problem to
incorporate node attributes into the ATF synchronization paradigm.

\subsection{Filesystems on directed acyclic graphs}

From the user's perspective a (hard or soft) \emph{link} between the nodes
$n$ and $n'$ is a promise, or a commitment, that the filesystem at and below
$n$ is exactly the same as at and below $n'$. In other words, the filesystem
acts as if the nodes $n$ and $n'$ in the filesystem skeleton were glued
together.

If the filesystem has many links and the links do not form loops, then after
this gluing the skeleton becomes a directed acyclic graph (DAG) with many
sources (the roots in the original skeleton). The gluing works in the other
direction, too: given any DAG with one or more sources, it can be
``unfolded'' into a forest. The paths of a tree-like filesystem can be
identified with the directed paths starting from a source, and two nodes are
``linked'' if the directed paths in the DAG lead to the same vertex. A DAG
vertex $v$ represents the collection of all nodes in the unfolded filesystem
which are determined by the directed paths in the DAG which lead to $v$. Two
nodes of this unfolded filesystem are \emph{equivalent}, written as
$n_1\simeq n_2$ if the corresponding directed DAG paths lead to the same
vertex. It is clear that $\simeq$ is an equivalence relation, and factoring
the filesystem by $\simeq$ yields the DAG. If every vertex in the DAG has
finite indegree, then the equivalence classes are also finite.

Operations on a DAG-based filesystem can be mimicked on the unfolded
filesystem by simply requesting that an operation performed on the DAG
vertex $v$ be done on all nodes represented by $v$. Similarly, a command set
$A$ on the unfolded filesystem corresponds to the command set $A/{\simeq}$
on the DAG-based filesystem if with every command $\sigma\in A$ all commands
$\simeq$-equivalent to $\sigma$ are also in $A$. We call these command sets
$\simeq$-invariant. Requesting all command sets to be $\simeq$-invariant,
Claims, Propositions and Theorems in this paper remain true. (Remark that in
this case the definition of a merger should require $M$ to be
$\simeq$-invariant.) Similarly, all algorithms continue to work, but they
must handle not commands but sets of $\simeq$-equivalent commands.
Consequently, time estimates are no longer valid. In summary, our results
and algorithms remain valid on filesystems based on arbitrary DAGs. It is an
open question whether the algorithms can be implemented in linear time in
the general case.

%======================================================================

\section*{Acknowledgment}

The work of the second author (L.Cs) was partially supported by the ERC
Advanced Grant ERMiD.
G\'abor Tardos' contribution for devising Algorithm \ref{alg:add-up}
is gratefully acknowledged.

%%%%%%%%%%%%%%%%%%%%%%%%%%%%%%%%%%%%%%%%%%%%%%%%%%%%%%%%%%%%%%%%%%%%%%%%%%%%%%%%
%\reftitle{References}
\bibliography{synchronizing}

\begin{thebibliography}{10}

\bibitem{AK08}
Micha{\l} Antkiewicz and Krzysztof Czarnecki.
\newblock {\em Design Space of Heterogeneous Synchronization}, pages 3--46.
\newblock Springer Berlin Heidelberg, Berlin, Heidelberg, 2008.

\bibitem{techradar23}
Desire Athow and Brian Turner.
\newblock Best file syncing solutions of 2023, 2023.
\newblock \url{https://www.techradar.com/best/best-file-syncing-solution}, Last
  accessed on 10 May, 2023.

\bibitem{BP98}
Sundar Balasubramaniam and Benjamin~C. Pierce.
\newblock What is a file synchronizer?
\newblock In William~P. Osborne and Dhawal~B. Moghe, editors, {\em {MOBICOM}
  '98, The Fourth Annual {ACM/IEEE} International Conference on Mobile
  Computing and Networking, Dallas, Texas, USA, October 25-30, 1998}, pages
  98--108. {ACM}, 1998.

\bibitem{BOS2023}
Novak Boškov, Ari Trachtenberg, and David Starobinski.
\newblock Enabling cost-benefit analysis of data sync protocols, 2023.

\bibitem{CC22a}
Elod~P. Csirmaz and Laszlo Csirmaz.
\newblock Data synchronization: A complete theoretical solution for
  filesystems.
\newblock {\em Future Internet}, 14(11), 2022.

\bibitem{Csi16}
Elod~Pal Csirmaz.
\newblock Algebraic file synchronization: Adequacy and completeness.
\newblock {\em CoRR}, abs/1601.01736, 2016.

\bibitem{DR10}
John Day-Richter.
\newblock What's different about the new {Google} {Docs}: {Making}
  collaboration fast, 2010.
\newblock
  \url{https://drive.googleblog.com/2010/09/whats-different-about-new-google-docs.html},
  Last accessed on 12 Jan, 2023.

\bibitem{graph-alg}
Shimon Even.
\newblock {\em Graph Algorithms}.
\newblock Cambridge University Press, USA, 2nd edition, 2011.

\bibitem{FQL12}
JiuLing Feng, XiuQuan Qiao, and Yong Li.
\newblock The research of synchronization and consistency of data in mobile
  environment.
\newblock In {\em 2012 IEEE 2nd International Conference on Cloud Computing and
  Intelligence Systems}, volume~02, pages 869--874, 2012.

\bibitem{K10}
Rusty Klophaus.
\newblock Riak core: Building distributed applications without shared state.
\newblock In {\em ACM SIGPLAN Commercial Users of Functional Programming}, CUFP
  '10, New York, NY, USA, 2010. Association for Computing Machinery.

\bibitem{knuth97}
Donald~E. Knuth.
\newblock {\em The Art of Computer Programming, Vol. 1: Fundamental
  Algorithms}.
\newblock Addison-Wesley, Reading, Mass., third edition, 1997.

\bibitem{LWJ13}
Zhenhua Li, Christo Wilson, Zhefu Jiang, Yao Liu, Ben~Y. Zhao, Cheng Jin,
  Zhi-Li Zhang, and Yafei Dai.
\newblock Efficient batched synchronization in dropbox-like cloud storage
  services.
\newblock In David Eyers and Karsten Schwan, editors, {\em Middleware 2013},
  pages 307--327, Berlin, Heidelberg, 2013. Springer Berlin Heidelberg.

\bibitem{LI21}
Erik Liu.
\newblock A {CRDT}-based file synchronization system.
\newblock Master's thesis, Department of Computer Science, 2021.

\bibitem{MOS2018}
Jakub~T. Mościcki and Luca Mascetti.
\newblock Cloud storage services for file synchronization and sharing in
  science, education and research.
\newblock {\em Future Generation Computer Systems}, 78:1052--1054, 2018.

\bibitem{NS16}
Agustina Ng and Chengzheng Sun.
\newblock Operational transformation for real-time synchronization of shared
  workspace in cloud storage.
\newblock In Stephan~G. Lukosch, Aleksandra Sarcevic, Myriam Lewkowicz, and
  Michael~J. Muller, editors, {\em Proceedings of the 19th International
  Conference on Supporting Group Work, Sanibel Island, FL, USA, November 13 -
  16, 2016}, pages 61--70. {ACM}, 2016.

\bibitem{NSC16}
Agustina Ng and Chengzheng Sun.
\newblock Operational transformation for real-time synchronization of shared
  workspace in cloud storage.
\newblock In {\em Proceedings of the 2016 ACM International Conference on
  Supporting Group Work}, GROUP '16, page 61–70, New York, NY, USA, 2016.
  Association for Computing Machinery.

\bibitem{PCS18}
Andrea Petroni, Francesca Cuomo, Leonisio Schepis, Mauro Biagi, Marco Listanti,
  and Gaetano Scarano.
\newblock Adaptive data synchronization algorithm for iot-oriented low-power
  wide-area networks.
\newblock {\em Sensors}, 18(11):4053, Nov 2018.

\bibitem{PMS09}
Nuno Pregui{\c c}a, Joan~Manuel Marques, Marc Shapiro, and Mihai Letia.
\newblock A commutative replicated data type for cooperative editing.
\newblock In {\em Proceedings of the 2009 29th IEEE International Conference on
  Distributed Computing Systems}, ICDCS '09, page 395–403, USA, 2009. IEEE
  Computer Society.

\bibitem{CRDT-overview}
Nuno~M. Pregui{\c{c}}a.
\newblock Conflict-free replicated data types: An overview.
\newblock {\em CoRR}, abs/1806.10254, 2018.

\bibitem{Qia04}
Yuechen Qian.
\newblock {\em Data synchronization and browsing for home environments}.
\newblock PhD thesis, Mathematics and Computer Science, 2004.

\bibitem{SLL11}
Bin Shao, Du~Li, Tun Lu, and Ning Gu.
\newblock An operational transformation based synchronization protocol for web
  2.0 applications.
\newblock In {\em Proceedings of the ACM 2011 Conference on Computer Supported
  Cooperative Work}, CSCW '11, page 563–572, New York, NY, USA, 2011.
  Association for Computing Machinery.

\bibitem{SPMB11}
Marc Shapiro, Nuno Pregui{\c c}a, Carlos Baquero, and Marek Zawirski.
\newblock A comprehensive study of {C}onvergent and {C}ommutative {R}eplicated
  {D}ata {T}ypes.
\newblock Technical Report 7506, INRIA, Inria-Centre Paris-Rocquencourt, jan
  2011.

\bibitem{CRDT-orig}
Marc Shapiro, Nuno~M. Pregui{\c{c}}a, Carlos Baquero, and Marek Zawirski.
\newblock Conflict-free replicated data types.
\newblock In Xavier D{\'{e}}fago, Franck Petit, and Vincent Villain, editors,
  {\em Stabilization, Safety, and Security of Distributed Systems - 13th
  International Symposium, {SSS} 2011, Grenoble, France, October 10-12, 2011.
  Proceedings}, volume 6976 of {\em Lecture Notes in Computer Science}, pages
  386--400. Springer, 2011.

\bibitem{syncpal-thesis}
Marius Shekow.
\newblock {\em Syncpal: a simple and iterative reconciliation algorithm for
  file synchronizers}.
\newblock PhD thesis, {RWTH} Aachen University, Germany, 2019.

\bibitem{SE98}
Chengzheng Sun and Clarence~A. Ellis.
\newblock Operational transformation in real-time group editors: Issues,
  algorithms, and achievements.
\newblock In Steven~E. Poltrock and Jonathan Grudin, editors, {\em {CSCW} '98,
  Proceedings of the {ACM} 1998 Conference on Computer Supported Cooperative
  Work, Seattle, WA, USA, November 14-18, 1998}, pages 59--68. {ACM}, 1998.

\bibitem{SJZ98}
Chengzheng Sun, Xiaohua Jia, Yanchun Zhang, Yun Yang, and David Chen.
\newblock Achieving convergence, causality preservation, and intention
  preservation in real-time cooperative editing systems.
\newblock {\em {ACM} Trans. Comput. Hum. Interact.}, 5(1):63--108, 1998.

\bibitem{TSR15}
Vinh Tao, Marc Shapiro, and Vianney Rancurel.
\newblock Merging semantics for conflict updates in geo-distributed file
  systems.
\newblock In Dalit Naor, Gernot Heiser, and Idit Keidar, editors, {\em
  Proceedings of the 8th {ACM} International Systems and Storage Conference,
  {SYSTOR} 2015, Haifa, Israel, May 26-28, 2015}, pages 10:1--10:12. {ACM},
  2015.

\bibitem{rsync}
Andrew Tridgell and Paul Mackerras.
\newblock The rsync algorithm.
\newblock 1996.

\bibitem{ZDA13}
Yupu Zhang, Chris Dragga, Andrea Arpaci-Dusseau, and Remzi Arpaci-Dusseau.
\newblock *-box: Towards reliability and consistency in dropbox-like file
  synchronization services.
\newblock In {\em Proceedings of the 5th USENIX Conference on Hot Topics in
  Storage and File Systems}, HotStorage'13, page~2, USA, 2013. USENIX
  Association.

\end{thebibliography}
\end{document}